\documentclass[sigconf, nonacm]{acmart}

\usepackage{booktabs} % For formal tables

\usepackage{setspace}
\usepackage{placeins}
\usepackage{afterpage}
\usepackage[utf8]{inputenc}
\usepackage{graphicx}
\usepackage{balance}
\usepackage{amsthm}
\usepackage{color}
\usepackage{caption}
\usepackage{amsmath}
\usepackage{tabu}
\usepackage{array}
\usepackage{booktabs}
\usepackage{threeparttable}
\usepackage{relsize}
\usepackage{physics}

\usepackage{algorithm}
\usepackage[noend]{algpseudocode}

\newlength{\maxwidth}
\newcommand{\algalign}[2]% #1 = text to left, #2 = text to right
{\makebox[\maxwidth][r]{$#1{}$}${}#2$}

\newcommand{\revision}[1]{\textcolor{black}{#1}}

\usepackage{float}
\usepackage{stfloats}
\usepackage{lipsum}
\usepackage[english]{babel}

\usepackage{comment}

\usepackage[T1]{fontenc}
\usepackage{mwe}    % loads »blindtext« and »graphicx«
\usepackage{subfig}

\let\oldnl\nl% Store \nl in \oldnl
\newcommand{\nonl}{\renewcommand{\nl}{\let\nl\oldnl}}% Remove line number for one line
\newtheorem{thm}{Theorem}
\newtheorem{problem}{Problem}

\theoremstyle{definition}
\newtheorem{defn}{Definition}
\def\Equal{\texttt{=}}
\newcommand\vldbdoi{XX.XX/XXX.XX}
\newcommand\vldbpages{XXX-XXX}
\newcommand\vldbvolume{16}
\newcommand\vldbissue{2}
\newcommand\vldbyear{2022}
\newcommand\vldbauthors{\authors}
\newcommand\vldbtitle{\shorttitle} 
\newcommand\vldbavailabilityurl{URL_TO_YOUR_ARTIFACTS}
\newcommand\vldbpagestyle{empty} 

\begin{document}
%\title{Models and Mechanisms for Fairness in Location Data Processing}
\title{Models and Mechanisms for Spatial Data Fairness}

%% The "author" command and its associated commands are used to define the authors and their affiliations.
\author{Sina Shaham}
%\authornote{Dr.~Trovato insisted his name be first.}
\affiliation{%
  \institution{Viterbi School of Engineering\\University of Southern California}
%  \streetaddress{P.O. Box 1212}
  \state{Los Angeles, California} 
  \country{USA}
}
\email{sshaham@usc.edu}

\author{Gabriel Ghinita}
%\authornote{Dr.~Trovato insisted his name be first.}
\affiliation{%
  \institution{College of Science and Engineering\\Hamad Bin Khalifa University\\Qatar Foundation, Doha, Qatar}
%  \streetaddress{P.O. Box 1212}
  %\state{Doha} 
  %\country{Qatar}
}
\email{gghinita@hbku.edu.qa}

\author{Cyrus Shahabi}
%\authornote{Dr.~Trovato insisted his name be first.}
\affiliation{%
  \institution{Viterbi School of Engineering\\University of Southern California}
%  \streetaddress{P.O. Box 1212}
 \state{Los Angeles, California} 
  \country{USA}
}
\email{shahabi@usc.edu}

\begin{abstract}
%Location data use became pervasive in the last decade due to the advent of mobile apps, as well as smart health and smart cities. At the same time, significant concerns have surfaced with respect to fairness in data processing. Individuals from certain population segments may be unfairly treated when being considered for loan or job applications, access to public resources, or other types of services. In the case of location data, fairness is an important concern, given that an individual's whereabouts are often correlated with sensitive attributes, e.g., race, income, education. 

Fairness in data-driven decision-making studies scenarios where individuals from certain population segments may be unfairly treated when being considered for loan or job applications, access to public resources, or other types of services.   In location-based applications, decisions are based on individual whereabouts, which often correlate with sensitive attributes such as race, income, and education.

While fairness has received significant attention recently, e.g., in machine learning, there is little focus on achieving fairness when dealing with location data. Due to their characteristics and specific type of processing algorithms, location data pose important fairness challenges.
We introduce the concept of {\em spatial data fairness} to address the specific challenges of location data and spatial queries. We devise a novel building block to achieve fairness in the form of {\em fair polynomials}. Next, we propose two mechanisms based on fair polynomials that achieve individual spatial fairness, corresponding to two common location-based decision-making types: {\em distance-based} and {\em zone-based}. Extensive experimental results on real data show that the proposed mechanisms achieve spatial fairness without sacrificing utility.
\end{abstract}

\maketitle

\pagestyle{\vldbpagestyle}
\begingroup\small\noindent\raggedright\textbf{PVLDB Reference Format:}\\
\vldbauthors. \vldbtitle. PVLDB, \vldbvolume(\vldbissue): \vldbpages, \vldbyear.\\
\href{https://doi.org/\vldbdoi}{doi:\vldbdoi}
\endgroup
\begingroup
\renewcommand\thefootnote{}\footnote{\noindent
This work is licensed under the Creative Commons BY-NC-ND 4.0 International License. Visit \url{https://creativecommons.org/licenses/by-nc-nd/4.0/} to view a copy of this license. For any use beyond those covered by this license, obtain permission by emailing \href{mailto:info@vldb.org}{info@vldb.org}. Copyright is held by the owner/author(s). Publication rights licensed to the VLDB Endowment. \\
\raggedright Proceedings of the VLDB Endowment, Vol. \vldbvolume, No. \vldbissue\ %
ISSN 2150-8097. \\
\href{https://doi.org/\vldbdoi}{doi:\vldbdoi} \\
}\addtocounter{footnote}{-1}\endgroup

\ifdefempty{\vldbavailabilityurl}{}{
\vspace{.3cm}
\begingroup\small\noindent\raggedright\textbf{PVLDB Artifact Availability:}\\
The source code, data, and/or other artifacts have been made available at %\url{https://infolab.usc.edu/DocsDemos/Sina_VLDB_2023.zip}.
\url{https://github.com/SinaShaham/c-Fair-Polynomials/tree/main}
\endgroup
}

\section{Introduction}
\label{sec:intro}

%Location data used extensively, including in decision making, ML, etc
In the past decade, location data became an integral part of many applications, e.g., mobile apps, smart health, smart cities. Individual locations are often used in creating user profiles,
%(determining the specific content and advertisements that a user is presented with), 
or as an input for various decision-making processes (e.g., machine learning), which may affect an individual's access to public resources, loans, etc.

%Significant bias may exist, due to the fact that locations are correlated with sensitive types. At the same type, a malicious entity can indirectly be biased by using locations.
It is already well-understood that location bias has significant effects on underprivileged communities, e.g., in the context of transportation and housing~\cite{PandeyC21}. Some public authorities designated entire geographical regions as economically-disadvantaged areas (EDA)~\cite{EDA}. However, while fairness has been studied recently in ML settings~\cite{barocas2017fairness} for generic data types~\cite{hajian2016algorithmic}, no specific solution studies fairness in spatial data processing.  \revision{
Addressing spatial data fairness presents two specific challenges:}

\revision{
(1) Location data may lead, intentionally or inadvertently, to exercise bias against individuals from disadvantaged backgrounds in a {\em stealth} fashion. While it is illegal to use race or ethnicity in a loan-granting or hiring decision, one may use location of current residence as input. Even though location may not seem sensitive, it may be used to discriminate against people of a certain ethnicity, as on many occasions people from the same ethnic group congregate in certain spatially-focused communities~\cite{kim2020spatial}. Similar concerns exist for income or education level, which often exhibit strong correlations with the location where an individual works, lives or travels~\cite{cantoni2020precinct}. Note that, this sort of discrimination may often occur inadvertently, as opaque ML algorithms automatically exploit certain correlations in location data (e.g., higher default rates in certain zipcodes), without realizing their fairness implications.}

\revision{
(2) Fairness is achieved through some data transformation designed to prevent, or limit, the amount of bias in processing. This causes loss of utility, whereby the result of processing can be sub-optimal compared to the result obtained on the original data. Achieving fairness requires some utility loss, and the emerging fairness-utility trade-off must be carefully considered when devising a fairness mechanism. In the case of location data, utility has specific formulations, which may impact results in a way that is unique to spatial query processing algorithms. Using generic fairness mechanisms devised for other types of data may lead to poor utility, as seen in~\cite{riederer2017price}. Therefore, it is desirable to design customized mechanisms for fighting bias in location data, such that the utility of spatial information is not significantly decreased. 
}

%Must achieve fairness for data
%To counter biases that are indirectly exercised through location data, it is important to characterize the concept of {\em location bias}, and extend existing definitions of fairness to the case of location data, which present specific challenges due to their intrinsic correlations to other attribute types, as well as their specific types of processing algorithms. 

In this paper, we introduce specific mechanisms targeted at providing {\em spatial fairness} while preserving data utility. We focus on the case of {\em individual} fairness~\cite{dwork2012fairness}, which is more difficult to achieve, but provides a higher level of fairness guarantees compared to its group-level counterpart. We provide specific definitions of location bias, and carefully characterize how location data can be used to exercise discriminatory decisions.

We introduce a novel construction called {\em fair polynomials} (Section 3.1) that can be used as building block within mechanisms for spatial fairness\footnote{While our focus is on geospatial data, some of our results can be extended for other types of data with continuous domains.}. We perform a detailed exploration of fair polynomials in order to understand their properties and the trade-off achieved between enforcing fairness and preserving data utility.

We identify two broad categories of scenarios where location bias occurs, and we define spatial fairness mechanisms for each: 
\begin{itemize}
    \item {\em Distance-based fairness} is relevant in location-based advertising and ride-hailing, where the dominant query type is nearest-neighbors (NN).
    In this setting, location bias occurs when individuals are impacted by their distance to a reference point. \revision{In location-based marketing, an algorithm may advertise special deals to customers that are nearby a newly-opened health food store. The specific coordinates of the customers may be less relevant for utility, and instead the distance to a landmark is the important factor (i.e., the proximity to the store is a good indicator of likelihood of visiting it). While this may be efficient, it can have fairness repercussions. For instance, if the algorithm chooses the first $100$ potential customers to reach based on distance, it is possible that only a rich neighborhood is covered. Customers from a poorer neighborhood that is adjacent to the rich one may never be selected, due to a slightly larger distance threshold. Figure~\ref{Fig: example dtr} illustrates this case. In this situation, we would like to ensure that individuals from poorer backgrounds also get the chance to benefit from special deals on healthy food, even though they are slightly farther away (i.e., avoid the hard decision boundary phenomenon).}
    \item {\em Zone-based fairness} is applicable in scenarios like gerrymandering, loan analysis or insurance pricing, where spatial range queries are the norm. In this case, we look at how to ensure spatial fairness with respect to coordinate values, instead of distances. This setting is broader, as it can provide fairness with respect to any reference point. Conversely, the amount of data utility sacrificed in the process may be higher. Figure~\ref{Fig: example coo} illustrates this point, where two individual homes are quoted significantly different insurance premiums due to their surrounding characteristics. Ideally, we would like the two residences to have similar premiums, given their physical proximity.
\end{itemize}

%Location data hide complex correlations, and fairness depends on interaction type, such as distance and query type; wee identify two main categories - examples, figures.

%Contributions
Our specific contributions are:
\begin{itemize}
    \item We identify the problem of bias in spatial data processing, and formalize the notion of spatial data fairness; 
    \item We introduce two definitions of spatial data fairness based on common interaction types, namely distance-based and zone-based fairness;
    \item We devise the novel concept of {\em fair polynomials}, which can be used as a building block to obtain mechanisms that achieve spatial fairness;
    \item We propose two mechanisms based on fair polynomials that enforce distance-based and zone-based fairness;
    \item We perform an extensive experimental evaluation on real datasets that shows the effectiveness of the proposed mechanisms, and investigates the fairness-utility trade-off.
\end{itemize}

%Outline
The rest of the paper is organized as follows: Section~\ref{Sec: System Model} formalizes the notions of bias and fairness for location data. Sections~\ref{Sec: Individual Fairness in DtR} and ~\ref{Sec: Individual Location Fairness of Coordinates} introduce our proposed distance-based and zone-based fairness mechanisms, respectively. We survey related work in Section~\ref{Sec: Related Work}. Section~\ref{Sec: Experimental Evaluation} reports the results of our experimental evaluation, followed by conclusions in Section~\ref{Sec: Conclusion}.

\begin{figure}[t]
    \centering
	\subfloat[\revision{Hard decision-boundary on distance to landmark leads to unfairness towards poor neighborhoods.}]{%
	\includegraphics[scale=.5]{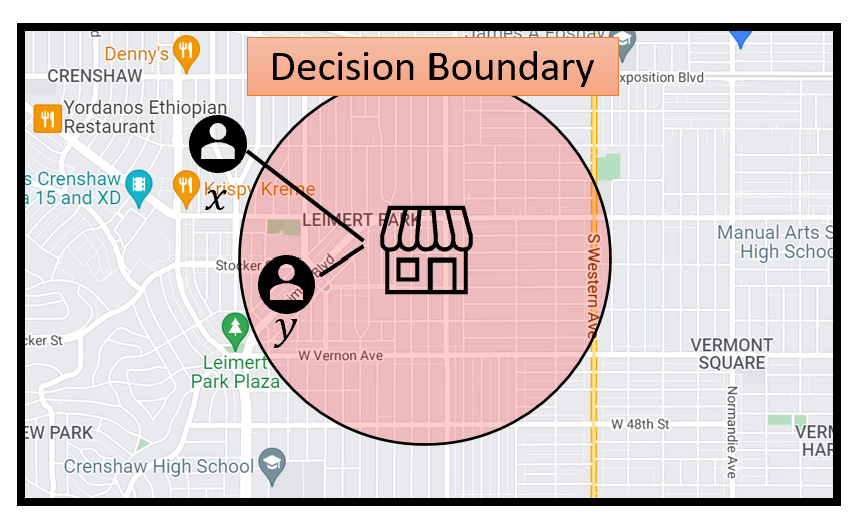}
	\label{Fig: example dtr}
	}
	%\qquad 
	%\qquad
	\vspace{-10pt}
	\subfloat[Despite close proximity, location $x$ is assigned a much higher insurance premium compared to $y$.]{%
	\includegraphics[scale=.5]{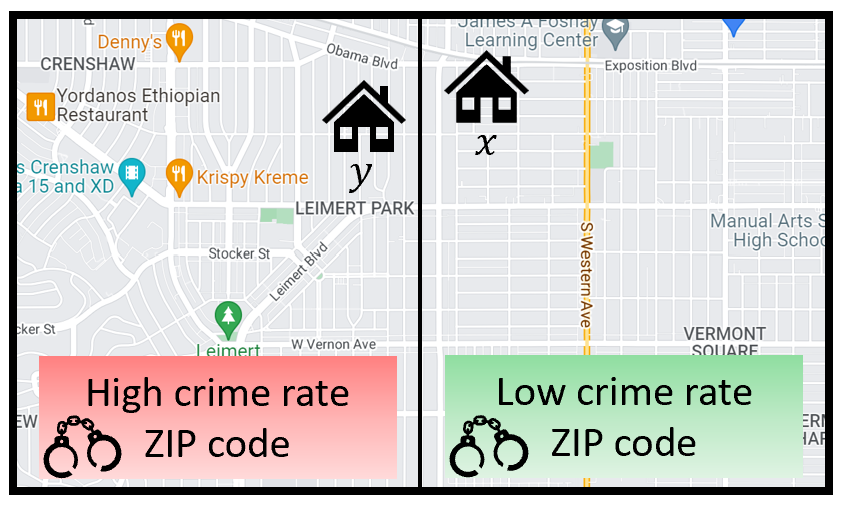}
	\label{Fig: example coo}
	}\vspace{-10pt}
	\caption{Examples of two common location bias scenarios.}
	\vspace{-10pt}
	\label{Fig: running example}
\end{figure}	

\section{System Model}\label{Sec: System Model}

\subsection{Location Bias}

The existence of bias, rooted in data or algorithms, is 
%a source of unfairness and is 
commonly used as a basis for reasoning on unfairness in decision-making. Several sources of bias have been identified in the literature, such as measurement bias~\cite{suresh2019framework} and behavioral bias~\cite{olteanu2019social}, many of which are intertwined. In this work, we formalize a type of bias that occurs due to location data. Location bias is formally defined as follows:
%in Definition~\ref{Def: location bias}, 

\begin{defn}[{\em Location Bias}]\label{Def: location bias}
Distortion or algorithmic bias generated based on locations of entities in the geospatial domain or their distances to reference points is referred to as {\em location bias}.
\end{defn}

%Location bias can be attributed to data as well as amplified or generated bias in the decision-making algorithm.
{\em Distortion} refers to the bias intrinsic to the data. Whereas algorithmic bias~\cite{baeza2018bias} refers bias that is generated by processing algorithms. A category of bias closely related to algorithmic bias is called data processing bias~\cite{olteanu2019social}, which occurs during data cleaning, enrichment, and aggregation. Our focus is on location bias sourced in processing algorithms. 
Location bias appears in a variety of applications where distances or locations may cause discrimination against individuals or groups. Consider the example in Fig.~\ref{Fig: Architecture} in which a store wants to contact nearby customers with daily offers. A widely known algorithm such as {\em nearest neighbors} may be used to decide which customers should be contacted. If fairness measures are not considered, and if the store is located in a rich neighborhood, customers who live in less privileged areas of the city may never be contacted, and will be unable to enjoy the special offers.

Such an algorithmic advantage towards a particular group or towards certain individuals is an example of {\em distance-based} location bias with respect to a reference point. The distances can be represented as a one-dimensional input feature; however, the source of location bias can be multidimensional, more commonly stored in two or more dimensions, e.g., as latitudes and longitudes. Consider classifying districts in a city as ``high risk'' or ``low risk'', where more police presence and government resources are allocated to locations with higher crime rates. A strict boundary between a low-risk and high-risk district dictates that nearby individuals placed oppositely across the border are treated differently, despite their proximity. The unfairness manifests in the number of police patrols, the tolerance of police officers to crime, the cost of insurance, and the approval of home improvement loans. 
%The scenario above might seem on a particular application basis, but the spread of its application can be seen in the prioritization of provisioning any common good or service.
%Regardless of one's opinion on what value is actually considered to be fair, whether it is fairer to have more police patrols in the area or less, one thing is certain that two 
To address location bias in classification tasks, i.e., to prevent individuals with similar features from being treated significantly different, i.e., {\em unfairly}, we first introduce the notion of location fairness in Section~\ref{Section: location fairness}, and then formulate the problem of achieving location fairness in Section~\ref{Problem Formulation}.

\subsection{Spatial Fairness Definition}\label{Section: location fairness}

There are two categories of fairness definitions, namely {\em group fairness} and {\em individual fairness}~\cite{dwork2018individual}. The former definition addresses the case where a group with certain features is treated statistically different compared to other groups. The latter approach focuses on treating individuals with similar features in a similar way. We adopt the use of individual fairness for spatial data, since {\em (i)} it provides higher fairness guarantees; {\em (ii)} it is more suitable for continuous domain features such as locations, in contrast to categorical features such as education, race, and gender; and {\em (iii)} location attributes tend to be dynamic, and hence more relevant on an individual basis, rather than for a group. To this end, we formally present the notion of {\em individual spatial fairness} in Definition~\ref{def: Individual location fairness}:
%(our definition is adapted from the individual fairness notion proposed in~\cite{dwork2012fairness}). 

\begin{defn} \label{def: Individual location fairness}
(Individual Spatial Fairness). Let $\mathcal{L} = \{ \boldsymbol{l_1},\boldsymbol{l_2},..., \boldsymbol{l_m} \}$ denote the set of individual locations that need to be classified over the output set $\mathcal{A}$, where $\boldsymbol{l_u}\in \mathbb{R}^k $. A randomized mapping $M:\; \mathcal{L} \rightarrow \Delta(\mathcal{A})$ satisfies individual spatial fairness iff for   every two locations $\boldsymbol{l_u},\boldsymbol{l_v}\in \mathcal{L}$ the $(D,d)$-Lipschitz constraint holds,
\begin{equation}
    D(M(\boldsymbol{l_u}),M(\boldsymbol{l_v}))\leq d(\boldsymbol{l_u},\boldsymbol{l_v})
\end{equation}
\end{defn}

Intuitively, the definition states that the evaluation process $M$ for two similar locations should yield similar outcomes. The definition relies on two key distance metrics, (1) similarity distance metric $d:\; V\times V \rightarrow \mathbb{R} $, measuring how similar individuals are, and (2) a distance metric $D(.)$ measuring the distance between outcome distributions. The former metric will be defined thoroughly for locations in the upcoming sections, as it is tailored to the specific location interaction type. The latter metric, on the other hand, is commonly defined as {\em total variation norm} or so-called {\em statistical distance}. %widely accepted in the literature.  
Given two probability distributions $P$ and $Q$ over outcome space $\mathcal{A}$, the statistical distance is calculated as 
\vspace{-5pt}
\begin{equation}\label{Equ: statistical distance}
    D(P,Q) = \dfrac{1}{2} \sum_{a\in \mathcal{A}} |P(a) - Q(a)|.
\end{equation}

%The intuition behind individual location fairness is that individuals located similarly or having similar distances from particular reference points should be treated similarly. 
Our focus in this work is on binary decision-making tasks, hence, the output space is given by $\mathcal{A} = \{0,1\}$. We assume a classifier modeled as a randomized mechanism $M: L \rightarrow \Delta(\mathcal{A})$ mapping individuals over outcomes, where $\Delta(\mathcal{A})$ denotes all possible distributions. Thus, the classification of an individual $\boldsymbol{l_u} \in \mathcal{L}$ over outcome space $\mathcal{A}$ is done according to the distribution of $M(\boldsymbol{l_u})$. To simplify notation, we assume function $M$ to return the likelihood of the positive outcome, i.e., $M(\boldsymbol{l_u}) = M(\boldsymbol{l_u}|a=1)$.

\newcommand{\rvec}{\mathrm {\mathbf {r}}} 
\begingroup
\begin{table}
\caption {Summary of notations.} 
\vspace{-10pt}
\centering
\begin{tabular}{>{\arraybackslash}m{2.5cm} >{\arraybackslash}m{5.8cm} }
\hline\hline
  Symbol  & Description\\    \hline
  $\mathcal{L} = \{ \boldsymbol{l_1},..., \boldsymbol{l_m} \}$ & Set of datapoints in $\mathbb{R}^k$ \\
  $l_i$ & Distance from $\boldsymbol{l_i}$ to reference point \\
  $m$   & Number of datapoints\\
  $||.||_p$ & $p$-norm distance\\
  $\mathcal{A}$ & Classification output domain\\
  $d(.)$ & Distance between datapoints\\
  $D(.)$ & Distance between distributions\\
  $M(\boldsymbol{l_i})$ & Likelihood score of location $\boldsymbol{l_i}$\\
\hline\hline
\end{tabular}
\label{tab:table1}
\vspace{-15pt}
\end{table}
\endgroup

\subsection{Problem Formulation}\label{Problem Formulation}

%We formulate two individual spatial fairness problems: {\em distance-based spatial fairness} and {\em zone-based spatial fairness}.
%In the following, we debut by providing the preliminaries of problem formulation and then the problem statements.  

Consider $m$ data points $\mathcal{L} = \{ \boldsymbol{l_1},\boldsymbol{l_2},..., \boldsymbol{l_m} \}$ located in a $k$-dimensional space ($\mathbb{R}^k$). Each location represents an individual. Associated attributes of data points are stored in a tabular format such that the $i^{th}$ row of the table is dedicated to $\boldsymbol{l_i}$ (\revision{as shown in Fig.~\ref{Fig: Architecture}).}

\begin{figure*}[t]
	\subfloat[Data collection]{%
	\includegraphics[scale=.7]{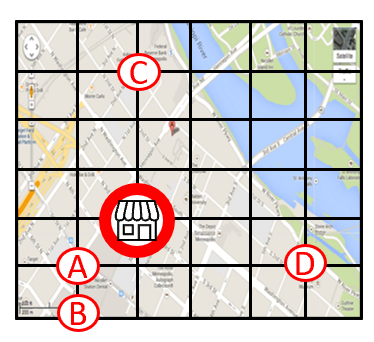}
	}
	\hfill
	\subfloat[User locations and classification scores]{%
	\includegraphics[scale=.6]{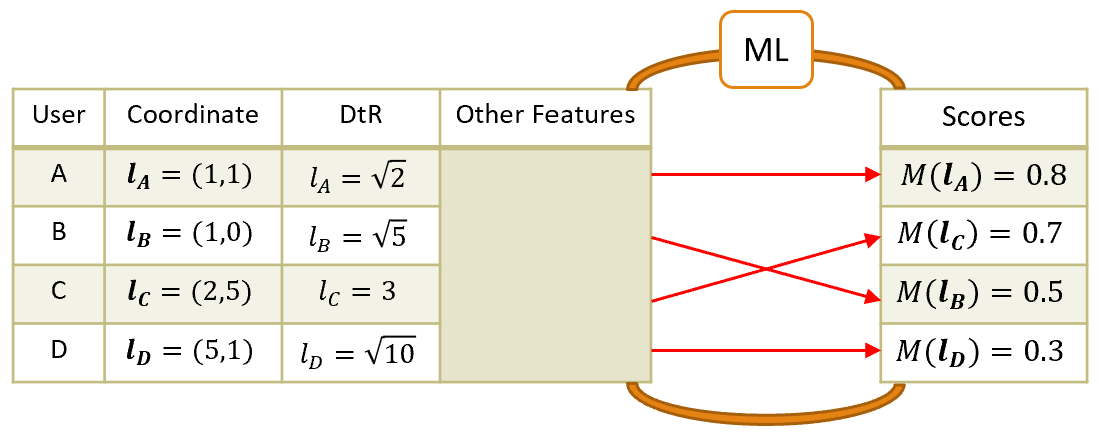}
	}
	\hfill
	\subfloat[$c$-fair polynomials]{%
	\includegraphics[scale=.7]{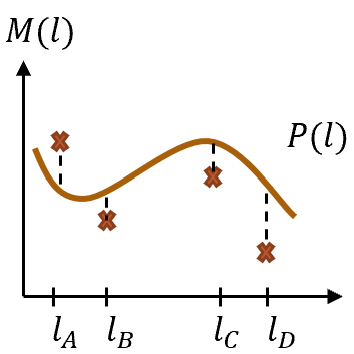}
	}
	\vspace{-10pt}
	\caption{An example of distance-based fairness problem.}
	\label{Fig: Architecture}
	\vspace{-10pt}
\end{figure*}

Attribute {\em Distance to Reference (DtR)} represents the distance from a reference point\footnote{\revision{The proposed approach can be extended to multiple reference points by using a composite cost metric, or an existing locality-sensitive hashing (LSH) approach that produces a single scalar distance value over multiple reference points.}} $\boldsymbol{R}$. We use as DtR metric the {\em Minkowski} distance of order $p$ ($p$-norm distance) defined for two data points $\boldsymbol{l_u} = (x_1,...,x_k)$ and $\boldsymbol{l_v} = (x_1',...,x_k')$ as 
\begin{equation}
    ||(\boldsymbol{l_u},\boldsymbol{l_v})||_p = \sqrt[p]{\sum_{i\Equal1}^{k}(x_i-x_i')^p}.
\end{equation}
The $u$-th entry of the DtR column is associated with datapoint $\boldsymbol{l_u}$ and entails a scalar $l_u = ||(\boldsymbol{l_u},\boldsymbol{R})||_p /\gamma$, where $\gamma =\underset{i=1...m}{max} ||(\boldsymbol{l_i},\boldsymbol{R})||_p $.  Constant $\gamma$ ensures the range of distances is [0,1] ($0\leq l_i\leq 1, \;\; \forall i\Equal 1...m$). The DtR column is shown with DtR metric set to $2$-norm. It is important to note that DtR is based on data representation, and it is not to be confused with the two distance metrics $d(.)$ and $D(.)$, which are crucial elements of individual fairness. 

\revision{As part of our system model, we assume a {\em classifier} performing decision-making on top of the data (this model aligns well with current location-based applications, where some machine learning is involved in data processing). Given an input individual $u$, the classifier $M(\boldsymbol{l_u})$ returns a likelihood score for that individual based on her location (e.g., likelihood of receiving a location advertisement). Scores are real values in the range of $0$ to $1$.}
%, indicating the likelihood of a classifier returning positive outcome $M(\boldsymbol{l_u}) = M(\boldsymbol{l_u}|a=1)$.
%In practice, the function $M$ could be any mapping function between data entries to values between zero and one, but we focus on a binary classifier with the scores interpreted based on the application. A score may show how creditworthy a user is to be granted a loan, how likely is the person to have been exposed to a disease, or in the case of the running example, how good of a candidate is the driver to pick up the client. 

The classifier output scores are shown in the last column of Fig.~\ref{Fig: Architecture}b.
%\footnote{User locations do not necessarily have to be used for training purposes.}. 
For example, user $A$ is located at the coordinate $\boldsymbol{l_A}$ with the calculated distance from the reference point of $\sqrt{1^2+1^2}=\sqrt{2}$ normalized to $\sqrt{2}/\sqrt{10}$, and the generated score of $M(\boldsymbol{l_A})=0.8$. Other features and attributes used in the model could be race, gender, education, etc. The two problems we seek to address to achieve individual fairness are defined as follows:

\begin{problem}({\em Distance-based Fairness})\label{Problem: DtR}
For a given location dataset $\mathcal{L}$ with the corresponding DtRs $\{l_1,..., l_m\}$, and a function $M: \mathcal{L}\rightarrow [0,1]$, devise a mechanism to enforce individual distance-based fairness $(D,d)$-Lipschitz constraints with respect to DtRs. 
\begin{equation}
    D (M(\boldsymbol{l_i}),M(\boldsymbol{l_j})) \leq d (l_i,l_j) \;\;\;\forall i,j\in [1,...,m]
\end{equation}
\end{problem}

\begin{problem}({\em Zone-based Fairness})\label{Problem: locations}
For a given location dataset $\mathcal{L} = \{ \boldsymbol{l_1},\boldsymbol{l_2},..., \boldsymbol{l_m} \}$, and a function $M: \mathcal{L}\rightarrow [0,1]$, devise a mechanism to enforce individual fairness $(D,d)$-Lipschitz constraints with respect to location coordinates. 
\begin{equation}
    D (M(\boldsymbol{l_i}),M(\boldsymbol{l_j})) \leq d (\boldsymbol{l_i},\boldsymbol{l_j}) \;\;\;\forall i,j\in [1,...,m]
\end{equation}
\end{problem}

Fairness mechanisms must inherently alter the output likelihood scores in order to achieve the fairness requirement. Hence, there is a cost for such an operation in terms of utility loss. Since we expect the output of a fairness mechanism to be used for a learning task, we choose as utility metric {\em  fitting error}, a widely accepted ML metric for output scores, formally presented in Definition~\ref{def: utility}.
\begin{defn} \label{def: utility}
(Utility). Let $\mathcal{B} :\; M \rightarrow M'$ be a mechanism that maps every likelihood score $M$ to a likelihood score in $M'$ given that $M,M':\; \mathcal{L} \rightarrow \Delta(\mathcal{A})$. The fitting error (utility) of $\mathcal{B}$ is: 
\begin{equation}
    \sqrt[2]{\dfrac{1}{m}\sum_{i\Equal1}^{m}(M(\boldsymbol{l_i})-M'(\boldsymbol{l_i}))^2}
\end{equation}
\end{defn}

\revision{
\noindent As an example, suppose that the output score is mapped to a constant $0.5$, i.e., $M'(\boldsymbol{l_i}) = 0.5$, $\forall i$. The utility loss can then be calculated as $\sqrt[2]{\dfrac{1}{4} (0.3^2 + 0.2^2 +0.2^2) } \approx 0.206 $. }

\section{Distance-based Fairness} \label{Sec: Individual Fairness in DtR}
We introduce a spatial fairness mechanism for
Problem~\ref{Problem: DtR}. Each data point is augmented with a DtR column representing a user's distance from the reference point. Therefore, the most natural similarity distance metric is $1$-norm.
\begin{equation}\label{Equation: d1}
    d (l_i,l_j) = |l_i- l_j|
\end{equation}
%The statistical distance is used to measure how different the output distributions are over outcomes. 
%Lemma~\ref{lemma simple calculation of D} introduces notation and the derivation of statistical distance for the classifier.

\begin{lemma}\label{lemma simple calculation of D}
Given the classifier output space of $\mathcal{A} = \{0,1 \}$, the statistical distance for every two individuals can be calculated as
\begin{equation}
    D (M(\boldsymbol{l_i}),M(\boldsymbol{l_j}))  =  |M(l_i) - M(l_j)|
\end{equation}
\end{lemma}

\begin{proof}
\begin{align}
    D (M(&\boldsymbol{l_i}),M(\boldsymbol{l_j}))  = \dfrac{1}{2} \sum_{a\in \{0,1\}} |P(a) - Q(a)|\\
    & =\dfrac{1}{2} (|M(\boldsymbol{l_i}) - M(\boldsymbol{l_j})| + |1- M(\boldsymbol{l_i})- (1- M(\boldsymbol{l_j}))|)\\
    &= |M(\boldsymbol{l_i}) - M(\boldsymbol{l_j})|
\end{align}
\end{proof}

%For example, for users $A$ and $B$ in Fig.~\ref{Fig: Architecture} the statistical distance can be derived as $D(\boldsymbol{l_A},\boldsymbol{l_B} )= 0.3$. Next, we describe the proposed mechanism to address Problem~\ref{Problem: DtR}.

\vspace{-10pt}
\subsection{Fair-Polynomials}

Despite the strong fairness guarantees provided by individual fairness, applying a large number of hard constraints has limited its practicality. A common mechanism for individual fairness is to define an application-specific optimization problem usually referred to as {\em vendor's utility function} and solve it while imposing individual fairness hard constraints. Unfortunately, two major issues arise with such an approach when applied to location data: (i) \revision{In existing approaches, e.g.,~\cite{riederer2017price}, a constraint optimization problem solver is used to alter input locations such that fairness requirements are met. The number of constraints grows quadratically with the number of data points, which makes their enforcement computationally prohibitive}; (ii) The definition of the utility function in most scenarios is not straightforward, confining applicable use cases. 

We devise the concept of {\em fair polynomials}, the intuition behind which is depicted in Figure~\ref{Fig: Architecture}c. A polynomial is efficiently fitted to the output scores of the classifier with a reasonably low {\em fitting error}. Fair polynomials no longer require enforcement of a large number of hard constraints: given a new data point, its corresponding fair value can be generated by evaluating the polynomial at that point.

\begin{defn}[$c$-fair Polynomials]~\label{Fair Polynomials for DtR}
     A single variable degree $n$ polynomial $P(x): \mathbb{R}\rightarrow \mathbb{R}$ is said to be $c$-fair if and only if for every two points $x$ and $y$ in its domain
     \begin{equation}~\label{Equ: individual fairness}
         |P(x)- P(y)| \leq c|x-y| 
     \end{equation}
\end{defn}

Given that a fair polynomial providing good estimates of likelihood scores exists, by one-to-one mapping (fitting) of the likelihood scores to polynomials, individual distance-based fairness can be achieved for every two data points. The constant $c\in [1,+\inf ]$ in $c$-fair polynomials aims to exploit the trade-off between utility and fairness. When $c=1$, the optimal individual location fairness is achieved but is usually associated with a higher loss in utility. When the value of $c$ grows larger, the fairness constraint is relaxed, leading to higher utility but lower fairness.

As the distance-based fairness problem only involves scalars, our focus is on single variable degree $n$ polynomials. We will extend to multi-variable polynomials to accommodate for multi-dimensional data points and address Problem~\ref{Problem: locations} in Section~\ref{subsection generalized}. In the following, we answer three central questions, (i) what is the sufficient condition for a polynomial to be fair, (ii) how to derive the coefficients of the polynomial by imposing individual fairness constraints, and (iii) how to determine the degree of a fair polynomial.

\subsection{Sufficient Condition for Fair Polynomials}

There are several families of polynomials that preserve individual fairness over the defined distances for DtR. For example, one such family of polynomials is $P(x) = cx^n/n$ which is proven in Lemma~\ref{Lemma: 1} to be a $c$-fair polynomial.
\begin{lemma}\label{Lemma: 1}
The polynomial $P(x) = cx^n/n$, is a $c$-fair polynomial for every two points $x,y\in [-1,1]$.
\end{lemma}

\begin{proof}
The proof can be derived by expanding the equation and applying triangle inequality, considering that $|x^iy^j|\leq 1, \forall i,j$.
\begin{align}
c|&\dfrac{x^n- y^n}{n}| = c|\dfrac{(x-y)(x^{n-1}+ x^{n-1}y+...+y^{n})}{n}|\leq \\
                 &  \dfrac{(c|x-y|)(|x^{n-1}|+ |x^{n-1}y|+...+|y^{n}|)}{n}|\leq                         c|x-y|
\end{align}
\end{proof}

\vspace{-10pt}
A fair polynomial must be flexible enough to reduce the error once the likelihood scores are fitted to the polynomial, and not every fair family of polynomials is a viable option. Consider the generic degree $n$ polynomial written as
\begin{equation}
    P(x) = a_0 +a_1 x+ ... +a_n x^n,
\end{equation}
where $a_i$ are real numbers. In Theorem~\ref{Theorem: sufficient condition 1}, we derive a sufficient condition for polynomials of order $n$ to preserve individual fairness.

\begin{thm}\label{Theorem: sufficient condition 1}
A sufficient condition for a single variable degree $n$ polynomial $P(x) =  \sum_{i\Equal 0}^{n} a_i x^i$ to be $c$-fair given that $a_i\in \mathbb{R}$ and $|x|\leq 1$ is to have,
\begin{equation}\label{Equation: dto cfair}
 \sum_{i\Equal 1}^{n} i | a_i| \leq c
\end{equation}
\end{thm}

\begin{proof}
Following the definition of individual location fairness:
\begin{align}\label{Equ: sufficiency condition}
    |P(x) - &P(y)|= | \sum_{i\Equal 1}^{n} a_i (x^i- y^i)| \leq  \sum_{i\Equal 1}^{n}| a_i (x^i- y^i)|\\
    & \leq  \sum_{i\Equal 1}^{n}(|a_i(x-y)|)(|x^{i-1}|+ |x^{i-2}y|+...+|y^{i}|)\\
    & \leq  \sum_{i\Equal 1}^{n}| a_i \times i  (x- y)|= |x-y|  \sum_{i\Equal 1}^{n} i | a_i|
\end{align}
The above inequality is true based on Jensen's inequality (and also, extended triangle inequality) as well as applying the result from Lemma~\ref{Lemma: 1}. Given that the inequality in Eq.~(\ref{Equation: dto cfair}) is satisfied, the polynomial is proven to be $c$-fair based on the definition. %In the above equation if $ \sum_{i\Equal 1}^{n} i | a_i| \leq 1$ the whole equation would be less than or equal to $|x-y|$. Thus, by definition, the polynomial is fair. 
\end{proof}

The theorem indicates that if likelihood scores generated by the model are fitted to a polynomial for which the coefficients are selected such that $ \sum_{i\Equal 1}^{n} i | a_i| \leq c$, then $c$-fairness is guaranteed for data entries. The sufficient condition in Theorem~\ref{Theorem: sufficient condition 1} can be used directly to learn $c$-fair polynomials, but the non-linearity existing in the constraint can result in higher computation complexity, as coefficients are unbounded. Theorem~\ref{Theorem: sufficient condition - start} addresses this problem by deriving linear constraints over coefficients.

\begin{thm}\label{Theorem: sufficient condition - start}

A sufficient condition for a $1$-variable $n$-th degree polynomial $P(x) =  \sum_{i\Equal 1}^{n} a_i x^i$ to be $c$-fair is to have: 
\begin{equation}%\label{Equation: mem}
|a_i|\leq \dfrac{6\times i\times c}{n(n+1)(2n+1)}   \;\;\;\forall \; i\in 1...n \;\;  (a_i\in \mathbb{R})
\end{equation}
\end{thm}

\begin{proof}
The bound on each $a_i$ value must allow for the maximum degree of freedom while fitting the likelihood scores. Therefore, the condition can be written as an optimization problem.

\begin{equation}
\begin{aligned}%\label{eq: optimization}
%& \underset{X}{\text{minimize}}
& \text{Minimize}
& & -\sum_{i=1}^n  a_i \\
& \text{Subject to}
%& & \sum_{j=0}^{i} n_j(j) = S, \; (i,j) \in \Omega, \\
& & \sum_{i=1}^n i | a_i| \leq c\\
%& & & l_i > 0,\;\; \forall i=1,..,n
\end{aligned}
\end{equation}

Writing the Lagrangian and applying the stationary condition of Karush–Kuhn–Tucker (KKT)~\cite{ghojogh2021kkt}, 
\begin{align}
    L(a_1,...,a_n,\lambda) &= -\sum_{i=1}^n  a_i - \lambda ( \sum_{i=1}^n i | a_i| - c)\\
    &\Rightarrow \dfrac{\partial L}{\partial a_i} = -1 - \lambda \dfrac{i}{a_i}=0 \rightarrow a_i = -\lambda i
\end{align}
$\lambda$ can be derived from complementary slackness to be
\begin{equation}
    \lambda \sum_{i=1}^n i^2 = c \rightarrow \lambda =  \dfrac{6c}{n(n+1)(2n+1)}
\end{equation}
Therefore, bounds on the coefficients are given as 
\begin{equation}
    |a_i|\leq \dfrac{6ic}{n(n+1)(2n+1)} \;\;\;\forall \; i\in 1...n
\end{equation}
\end{proof}

\subsection{Derivation of Fair Polynomials}
%So far, we derived sufficient conditions for a polynomial to be fair. Next, we show how the coefficients of a fair polynomial can be derived based on the likelihood scores under the sufficiency condition. 
\revision{We employ a simple ML model to compute polynomial coefficients, where each location distance represents a training sample used to fit the likelihood scores to a polynomial. The training set can be assembled by choosing at random a number of locations from the same data domain as the application (e.g., residential locations within a city). In cases where the target user population is already known (e.g., the coordinates of customers for a store that is using location-based advertising), this user set can be directly used for training. In the following, to simplify notation, we assume the latter.} 
For a given training input $l_i$, the polynomial output is derived as
\begin{equation}
    P(l_i) = a_0 +a_1 l_i+ ... +a_n l_i^n,
\end{equation}
%Our variables that we would like to learn are $a_i$ in this equation to
We denote the matrix of all training examples as 
\begin{equation}\nonumber
L=
\begin{bmatrix}
1 & l_1 & l_1^2 &... & l_1^n\\
1 & l_2 & l_2^2 &... & l_2^n\\
. & . & . &.& .\\
1 & l_n & l_n^2 &... & l_m^n
\end{bmatrix}.
\end{equation}
 Recall that $m$ is the number of training examples, and the variables that we learn are the $a_i$s, which define the fair polynomial fitted to the data. The vector of coefficients can be written as 
\begin{equation}
\boldsymbol{a}^T = [a_0, a_1,a_2,..., a_n]
\end{equation}
and the likelihood scores are vectorized as  
\begin{equation}
    \boldsymbol{b}^T = [M(l_1),M(l_2),...,M(l_m)]
\end{equation}
The convex optimization problem to learn $\boldsymbol{a}$ is formulated as: 

\begin{equation}
\begin{aligned}\label{eq: optimization}
%& \underset{X}{\text{minimize}}
& \text{Minimize}
& & || L\boldsymbol{a} - \boldsymbol{b} ||_2\\
& \text{Subject to}
%& &  \sum_{i\Equal 1}^{n}_{j=0}^{i} n_j(j) = S, \; (i,j) \in \Omega, \\
& & |a_i|\leq \dfrac{6\times i\times c}{n(n+1)(2n+1)}\\
%& & & l_i > 0,\;\; \forall i=1,..,n
\end{aligned}
\end{equation}
This is equivalent to the least square problem with linear constraints and can be solved efficiently with algorithms such as Trust Region Reflective~\cite{conn2000trust} and Bounded-variable least-squares~\cite{stark1995bounded}, with complexity linear to the order of the polynomial $n$. Existing work~\cite{riederer2017price} requires enforcing $O(m^2)$ hard constraints ($m>>n$).

\revision{The selection of the polynomial degree $n$ can be conducted based on a trial and error methodology. The optimal degree is the one that results in the minimum variance of error between likelihood scores and their corresponding values on the polynomial. Formally, let $e_i$ denote the error between $M(\boldsymbol{l_i})$ and $P(l_i)$, i.e.,}
\begin{equation}
    e_i = |M(\boldsymbol{l_i})- P(l_i)|
\end{equation}
Then, the value of $n\geq 1$ is selected such that
\begin{equation}
%\underset{n}{\text{Minimize}}\;\;\;\; (\sum_{i\Equal 1}^m e_i^2)/ (m-n-1)
n=argmin(\sum_{i\Equal 1}^m e_i^2)/ (m-n-1)
\end{equation}

\section{Zone-based Fairness}\label{Sec: Individual Location Fairness of Coordinates}

%In the previous section, we showed how $c$-fair polynomials represent a feasible mechanism for achieving individual fairness in the DtR setting. Next, 
%We extend fair polynomials to preserve individual zone-based fairness. 
%Zone-based fairness ensures the deviation of outcome scores achieved for locations in a classification task should not differ significantly for . 
Revisiting the example in Figure~\ref{Fig: Architecture}, suppose that two users $A$ and $B$ both apply for a home improvement loan, and despite living in close proximity, one is categorized in an underdeveloped area and the other in a developed region due to geographic segmentation. A bank applies a classifier to decide whether an applicant should be granted a loan. The applicant whose home is in the underdeveloped category might be disadvantaged, as the location category can significantly impact the output of the classifier. Individual fairness argues that if two users are located close to each other, their output likelihood scores should not differ significantly. 

In the distance-based fairness case, a single variable $c$-fair polynomial can fit output scores due to scalar distances. For multi-dimensional data points, Definition~\ref{Fair Polynomials for DtR} is no longer directly applicable. To address this problem, we extend the definition to multivariate polynomials to achieve individual fairness for higher dimensional data points. The number of variables involved in fair polynomials is equal to the dimensionality of data points ($k$).

Three key variables are involved in finding an efficient family of fair polynomials that can fit the output likelihood scores with low utility loss: {\em (i)} dimensionality of data $k$; {\em (ii)} the distance metric $d(.)$ and {\em (iii)}  fair polynomial degree $n$. The individual fairness problem can be characterized with respect to these criteria as follows:
%We consider the following generic cases for which the first item was addressed in Section~\ref{Sec: Individual Fairness in DtR} and the remaining will be illustrated in this section. 
%\textcolor{blue}{}{
\revision{
\begin{itemize}
    \item One-dimensional data representation (scalars), $1$-norm distance, flexible order polynomial. This corresponds to the distance-based fairness case. 
    %Common applications of the DtR problem include, (I) distance to a single origin, (II) distance of the closest reference point, and (III) the traveled distances stored as scalars. 
    \item $2$-Dimensional data representation, $2$-norm distances; order 1 polynomial. This is the most common scenario for locations where the attribute columns include 2D coordinates, and the fairness must be achieved with respect to Euclidean distance between individuals. 
    \item $k$-Dimensional data representation, $2$-norm distances; order 1 polynomial.
    \item $k$-Dimensional data representation, $p$-norm distances; order 1 polynomial.
    \item $k$-Dimensional data representation, $p$-norm distances; flexible order polynomial.
\end{itemize}
}

We formulate and derive the sufficiency condition to guarantee individual fairness for each mentioned scenario. The optimization problem in Eq.~\eqref{eq: optimization} is formulated with the derived constraints. We omit the vectorization process for conciseness. For several of the proofs used in this section, we make use of Generalized Titu's Lemma provided in Lemma~\ref{Lemma: Titu}.

\begin{lemma}[\em Generalized Titu's Lemma]\label{Lemma: Titu}
Let $m$ be an integer greater than or equal to $2$, $a_i^m$ a non-negative real number, and $x_i$ a positive real number. Then,
\begin{equation}
    n^{m-2} \sum_{i\Equal 1}^{n} \dfrac{a_i^m}{x_i}\geq \dfrac{(\sum_{i\Equal1}^n a_i)^m}{\sum_{i\Equal1}^n x_i}
\end{equation}

\end{lemma}
\begin{proof}
Proof is in Appendix~\ref{Section: appendix} of our extended version\footnote{ https://github.com/SinaShaham/c-Fair-Polynomials/tree/main}.
\end{proof}

\subsection{$2$-norm, $2$ dimensional data, Order $1$ polynomial}\label{subsection generalized}

For higher dimensional data, the most common scenario happens when data points are in 2D, and the order of the polynomial is one. In practice, data points represent coordinates of locations on the map. Consider two locations $\boldsymbol{l_1} = (x_1,x_2)$ and $\boldsymbol{l_2} = (x_1',x_2')$ in $\mathbb{R}^2$, where $x_1$ and $x_1'$ are the $x$-axis coordinates, while $x_2$ and $x_2'$ denote $y$-axis coordinates. To achieve individual fairness with respect to locations, the hard Lipschitz constraints dictate that:

\begin{equation}
    D (M(\boldsymbol{l_i}),M(\boldsymbol{l_j})) \leq d (\boldsymbol{l_i},\boldsymbol{l_j})\;\; \forall i,j \in 1...m 
\end{equation}

The distance between distribution scores, i.e., $D(.)$, is calculated as before based on Equation (\ref{lemma simple calculation of D}) and the distance between locations is the $2$-norm of data points (Euclidean distance), calculated as:

\begin{equation}
    d (\boldsymbol{l_i},\boldsymbol{l_j}) = \sqrt[2]{(x_1-x_1')^2+ (x_2-x_2')^2}
\end{equation}

We start by showing how a fair-polynomial can be derived for the Euclidean similarity distance. Then, we relax the assumptions and generalize the approach for arbitrary distance norms as well as $n$-dimensional data points. Recall that as location data are stored in $2$D, the fair polynomial consists of two variables. The generalized definition of fair polynomials for order $n$ polynomials and $k$ dimensional data is provided in Definition~\ref{Def: generalized fair polynomials}.

\begin{defn}[Generalized $c$-Fair Polynomial] \label{Def: generalized fair polynomials}
     The polynomial $P(x_1,x_2,...,x_k): \mathbb{R}^k\rightarrow \mathbb{R}$ with real coefficients is $c$-fair iff for every two points $x = (x_1,x_2,...,x_m)$ and $x' =(x_1',x_2',...,x_m') $ in its domain
     \begin{align}~\label{Equ: individual fairness multi-variable}
         |P(x_1,x_2,...,x_m)- P(x_1',x_2',&...,x_m')| \leq \\&c \times d(x,x') = c \times|| (x,x') ||_p\nonumber
     \end{align}
\end{defn}

In the case of $2$-dimensional locations and Euclidean distance, fair polynomials imply that for every two locations $\boldsymbol{l_1} = (x_1,x_2)$ and $\boldsymbol{l_2} = (x_1',x_2')$, we must have,

\begin{equation}
|P(x_1,x_2)- P(x_1',x_2')| \leq c \times \sqrt[2]{(x_1-x_1')^2+ (x_2-x_2')^2}
\end{equation}
Where the polynomial is denoted by
\begin{equation}
    P(x_1,x_2) = a_0 + a_1 x_1 + a_2 x_2
\end{equation}
The goal is to learn the coefficients $a_i$ such that the polynomial $P(.)$ can model the output scores $M(.)$ and preserve fairness with respect to Euclidean distance. Theorem~\ref{Theorem: sufficient condition2} provides the sufficiency condition for a two-variables order one polynomial to be fair.

\begin{thm}\label{Theorem: sufficient condition2}

A sufficient condition for a $2$-variable first degree polynomial $P(x_1,x_2)= a_0 + a_1 x_1 + a_2 x_2$ to be $c$-fair defined over $2$-norm similarity distance is to have: 
\begin{equation}\label{Equation: mem}
| a_1|,|a_2| \leq c/\sqrt{2}    \;\;\;\;  (a_1\, a_2\in \mathbb{R})
\end{equation}
\end{thm}

\begin{proof}
On the one hand, based on Lemma~\ref{Lemma: Titu}, a lower bound for Euclidean distances can be written as
\begin{align}
    d (\boldsymbol{l_i},\boldsymbol{l_j})= &\sqrt[2]{(x_1-x_1')^2+ (x_2-x_2')^2}\geq \\
    &\sqrt[2]{(|x_1 - x_1'|+ |x_2-x_2'|)^2/2}
\end{align}
On the other hand, for the polynomial one can write

\begin{align}
    |P(x_1,x_2)- P(x_1',x_2')| &= |a_1 (x_1 - x_1')+ a_2(x_2 - x_2')|\\
    &\leq |a_1| |(x_1 - x_1')|+ |a_2||(x_2 - x_2')|
\end{align}
By combining the two equations the sufficiency condition in Equation (\ref{Equation: mem}) can be derived from the following inequality
\begin{align}
     |a_1| |(x_1 - x_1')|+ &|a_2||(x_2 - x_2')|\leq \\ &c \times (|x_1 - x_1'|+ |x_2-x_2'|)/\sqrt{2}\nonumber
\end{align}
\end{proof}

The above theorem indicates that if the coefficients of polynomials fitted to data are chosen such that $| a_1|,|a_2| \leq c/\sqrt{2} $, fairness is guaranteed for every two locations in the domain. The sufficiency condition for first degree polynomials is generalized for $k$-dimensional data points in space in Theorem~\ref{Theorem: sufficient condition - gen} (the number of variables in the polynomial is equal to the number of dimensions).

\begin{thm}\label{Theorem: sufficient condition - gen}
A sufficient condition for a $k$-variable first degree polynomial $P(x_1,...,x_k)= a_0 +\sum_{i\Equal 1}^k a_i x_i $ defined over $2$-norm similarity distance to be $c$-fair is: 
\begin{equation}
| a_i| \leq c/\sqrt[2]{k},\;\;\;\; \forall i=1...k,  \;\;\;\; (a_i\in \mathbb{R})
\end{equation}
\end{thm}

\begin{proof}
Please see proof in Appendix~\ref{Section: appendix} of our extended version.
\end{proof}

\subsection{$p$-norm, $k$ dimensional, Order $1$ polynomial}

%So far, we have provided the sufficiency condition for $2$-norm in $k$-dimensional space with order $1$ polynomials. 

We relax the similarity metric for arbitrary $p$-norm distance, calculated for data points $\boldsymbol{l_i} = (x_1,...,x_k)$ and $\boldsymbol{l_j} = (x_1',...,x_k')$ as
\begin{equation}
    d (\boldsymbol{l_i},\boldsymbol{l_j}) = \sqrt[p]{\sum_{q\Equal1}^{k}(x_q-x_q')^p}
\end{equation}

\begin{thm}%\label{Theorem: sufficient condition4}
A sufficient condition for a $k$-variable first degree polynomial $P(x_1,...,x_k)= a_0 +\sum_{i\Equal 1}^k a_i x_i $  defined over $p$-norm similarity distance to be $c$-fair is: 
\begin{equation}
| a_i| \leq c/\sqrt[p]{k^{p-1}},\;\;\;\; \forall i=1...k 
\end{equation}
\end{thm}

\begin{proof}
Based on generalized Titu's Lemma, we have on the one hand a lower bound for Euclidean distances:
\begin{align}\label{Equ: p-norm inequality}
    d (\boldsymbol{l_i},\boldsymbol{l_j})= &\sqrt[p]{\sum_{q\Equal1}^{k}(x_q-x_q')^p}\geq \\
    &\sqrt[p]{(\sum_{q\Equal1}^{k}|x_q-x_q'|)^p/k^{p-1}} = (\sum_{q\Equal1}^{k}|x_q-x_q'|)/\sqrt[p]{k^{p-1}}\nonumber
\end{align}
On the other hand, for the polynomial one can write \begin{align}
    |P(x_1,...,x_k)- P(x_1',...,x_k')| &= |\sum_{q\Equal1}^{k} a_q (x_q - x_q') |\\
    &\leq \sum_{q\Equal1}^{k} |a_q| |(x_q - x_q')|
\end{align}
Combining the two, we obtain:
\begin{equation}
     \sum_{q\Equal1}^{k} |a_q| |(x_q - x_q')| \leq \sum_{q\Equal1}^{k}|x_q-x_q'|/\sqrt[p]{k^{p-1}}
\end{equation}
The inequality is satisfied when $|a_q|\leq 1/\sqrt[p]{k^{p-1}}$.
\end{proof}

\subsection{$p$-norm, $k$ dimensional, Order $n$ polynomial}

So far, the sufficiency condition for $c$-fair polynomials was derived for arbitrary norms in $k$-dimensional space based on order $1$ polynomials. Moreover, for distance-based fairness, $c$-fair polynomials were developed for $1$-dimensional distance using arbitrary degree $n$ polynomial.  This subsection provides the theoretical background for the generalized scenario in which the location data are in $k$ dimensions with the norm set to $p$, and degree $n$ polynomials. 

Although by increasing the degree of polynomials, a better fit to likelihood scores can be achieved, the existence of monomials in which multiple variables are involved leads to complexity in the derivation of sufficiency conditions. To address this, we assume that the monomials in the multivariable polynomial consist of only a single variable. Making such an assumption comes with the cost of utility loss; however, it greatly reduces the complexity of the generic case. We assume that the degree $n$ polynomial is expressed as the summation of $k$ univariate polynomials. 
\begin{equation}\label{equ: components}
    P(x_1,x_2,...,x_k) = \sum_{i\Equal 1}^k P_i(x_i),
\end{equation}
where $P_i(x_i) = \sum_{j\Equal 1}^{n} a_{ij} x_i^j$ is a degree $n$ univariate polynomial with its input being $x_i$, the $i$th variable in the original polynomial. The assumption helps to remove existence of monomials with multiple variables such as $x_i^3x_j^2x_k^3$ and to simplify derivation of location fairness sufficiency conditions provided in Theorem~\ref{Theorem: sufficient condition4}.

\begin{thm}\label{Theorem: sufficient condition4}
A sufficient condition for a $k$-variable $n$-th degree polynomial $P(x_1,x_2,...,x_k) = \sum_{i\Equal 1}^k P_i(x_i)$ to be $c$-fair defined over $p$-norm similarity distance is to have: 

\begin{equation}
|a_{ij}|\leq \dfrac{6\times j\times c\sqrt[p]{k^{p-1}}}{n(n+1)(2n+1)},\;\;\;\; \forall i=1...k\; \& \; \forall j=1...n
\end{equation}
\end{thm}
\begin{proof}
We write the equation in its component form shown in Equation~\ref{equ: components}. An upper bound for $P_i(x_i)$ can be derived as, 
\begin{align}
    |P(x_1,...&,x_k)- P(x_1',...,x_k')| = |\sum_{i\Equal 1}^k P_i(x_i) - \sum_{i\Equal 1}^k P_i(x_i') |\\
    &= |\sum_{i\Equal 1}^k (P_i(x_i) - P_i(x_i')) |
    \leq \sum_{i\Equal 1}^k | P_i(x_i) - P_i(x_i') |
\end{align}
An upper bound for the sub-terms for all $j=1...k$ can be derived as

\begin{align}
    |P_i(x_i) - &P_i(x_i')|= | \sum_{j\Equal 1}^{n} a_{ij} (x_i^j- x_i'^j)| \leq |x_i- x_i'|  \sum_{j\Equal 1}^{n} j | a_{ij}|
\end{align}
The above inequality is derived based on Eq.~(\ref{Equ: sufficiency condition}) setting $|x_i|\leq 1$. We also use the lower bound derived in Eq.~(\ref{Equ: p-norm inequality})
\begin{equation}
   c\times d (\boldsymbol{l_u},\boldsymbol{l_v}) \geq c\sum_{q\Equal1}^{k}|x_q-x_q'|/\sqrt[p]{k^{p-1}}
\end{equation}
Putting the derived upper bound and lower bound together, the following inequality is satisfied, and $c$-fairness is guaranteed. 
\begin{align}
    |x_i- x_i'|  &(\sum_{j\Equal 1}^{n} j | a_{ij}|) \leq c|x_j-x_j'|/\sqrt[p]{k^{p-1}}\\
    &\sum_{j\Equal 1}^{n} j | a_{ij}| \leq c/\sqrt[p]{k^{p-1}}
\end{align}
By applying the method used in DtR, the bounds are linearized to,
\begin{equation}%\label{Equation: mem}
|a_{ij}|\leq \dfrac{6\times j\times c}{n(n+1)(2n+1)\sqrt[p]{k^{p-1}}}   \;\;\;\forall \; i\in 1...n \;\;  (a_i\in \mathbb{R})
\end{equation}

\end{proof}

It is worth noting that the generated scores by fair polynomials can result in values greater than one or less than zero. In such scenarios, the values are suppressed to one and zero, respectively. It is straightforward to prove that the suppression process does not violate the individual fairness constraints and leads to higher utility.

\section{Related Work}\label{Sec: Related Work}

%With the growing number of decision-making tasks conducted by Machine Learning (ML), data bias may be introduced in decision-making systems, inadvertently or not.  We review some of the existing fundamental fairness notions and the mechanisms to achieve them, particularly prior works related to geospatial data.\\

%{\em Existing Notions of Fairness.} 
Fairness notions can be grouped into two broad categories of Group Fairness and Individual Fairness~\cite{dwork2012fairness}. In group fairness, a protected attribute of the dataset, such as race or gender, which is considered to be critical in decision-making outcomes, partitions individuals into groups. The ML model used for a decision-making task on the dataset is considered to be fair if it achieves some statistical measure across groups. A few of the key statistical measures include statistical parity~\cite{kusner2017counterfactual}\cite{dwork2012fairness}, equalized odds~\cite{hardt2016equality}, treatment equality~\cite{berk2021fairness}, and test fairness~\cite{chouldechova2017fair}. Individual fairness aims to give similar predictions to similar outcomes, focusing on fairness for individuals as opposed to groups. Group fairness notions are generally weaker than individual fairness notions~\cite{kim2019preference}. Despite higher fairness guarantees provided by individual fairness and fragility of group fairness notions, group fairness notions are widely studied in the literature due to their easier enforcement~\cite{mehrabi2021survey}. Only a handful of approaches exist in the literature to achieve fairness in the geospatial domain.
%, which are reviewed in the following.\\

%{\em Achieving Individual Fairness.} 
%In spite of higher fairness guarantees provided by individual fairness, the lack of an efficient mechanism to enforce it has limited its practicality. 
The current state-of-the-art approach to enforce individual location fairness is to define a linear loss function once the likelihood scores are generated and solve optimization under individual fairness Lipschitz constraints. Let $I$ be an instance of our problem consisting of a metric $d:\; \mathcal{L}\times \mathcal{L} \rightarrow \mathbb{R}$, and a loss function $L:\; \mathcal{L}\times A\rightarrow \mathbb{R}$, the optimization problem is defined as,

\begin{align}
    opt(I) = & \underset{\{M(x) \}_{x\in L}}{min}\;\; \underset{x\sim L}{\boldsymbol{E}} \;\; \underset{a\sim M(x)}{\boldsymbol{E}} \;\; L(x,a)\\
    & \text{Subject to:  } \forall x,y\in : \;\; D(M(x),M(y))\leq d(x,y)\\
    &\forall x\in L: \;\; M(x)\in \Delta (A)
\end{align}

One can see that the number of constraints in this mechanism grows quadratically with the number of individuals, imposing a large computational complexity on the system. The authors in~\cite{riederer2017price} formulate the loss function for location-based advertisements in social media. Locations visited on the map are shown as binary strings, and a classifier is used to predict whether a user should receive a targeted advertisement. Moreover, not directly related to locations, but for general purpose advertisement and auctions, individual fairness is applied in~\cite{dwork2018individual}. Another application over which the loss function has been defined is achieving individual fairness in ranking and recommendation systems~\cite{pitoura2021fairness}. In ranking systems, the amount of unfairness with respect to individuals is measured after ranking, and a loss function aims to reorder ranking such that the amount of individual unfairness is minimized~\cite{biega2018equity}.

Several attempts have also been made to apply the individual fairness notion for clustering datapoints in Cartesian space. The notion in~\cite{kleindessner2020notion} defines clustering conducted for a point in space as fair if the average distance to the points in its own cluster is not greater than the average distance to the points in any other cluster. The authors in~\cite{mahabadi2020individual} focus on defining individual fairness for $k$-median and $k$-means algorithms. Clustering is defined to be individually fair if every point expects to have a cluster center within a particular radius. To the best of our knowledge, no work has directly defined individual fairness with respect to locations.

\begin{comment}
\begin{figure}[t]
%\centering
\raggedright
\includegraphics[scale=.45]{Figures/Average_Price_per_Kilometer.pdf}
%\hspace{1em}
\centering
\vspace{-10pt}
\caption{Taxi Fares Distribution over Time of Day}
%\label{Fig: chicago distribution}
\vspace{-10pt}
\end{figure}
\end{comment}

\begin{figure}[!ht]
\centering
	\subfloat[Percentage of Unfairness versus $c$.\label{fig: NewYork1}]{%
	\includegraphics[scale = 0.45]{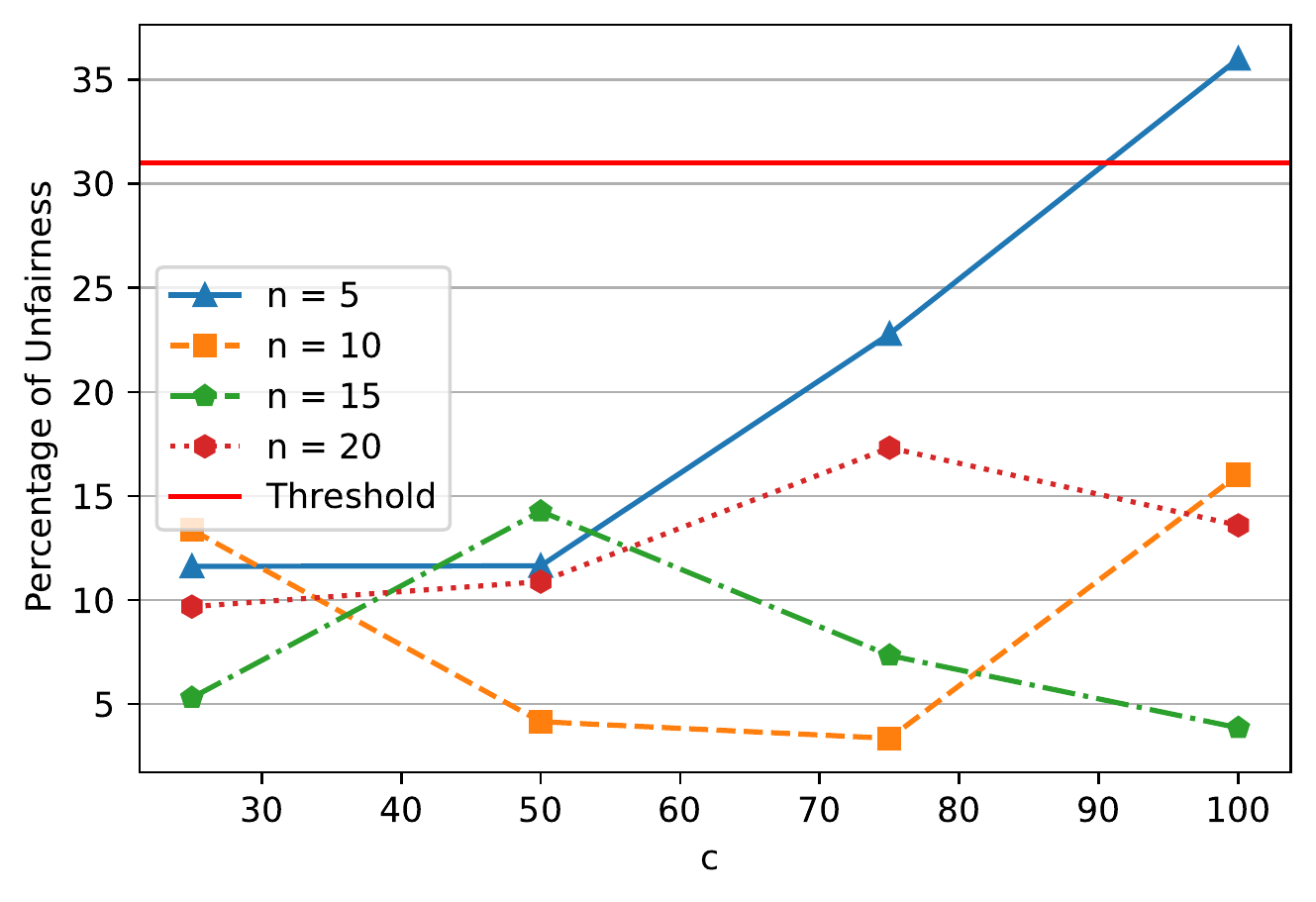}
	}
	\vspace{-10pt}
		\hfill
	\subfloat[Percentage of Unfairness versus $n$.\label{fig: NewYork2}]{%
		\includegraphics[scale = 0.45]{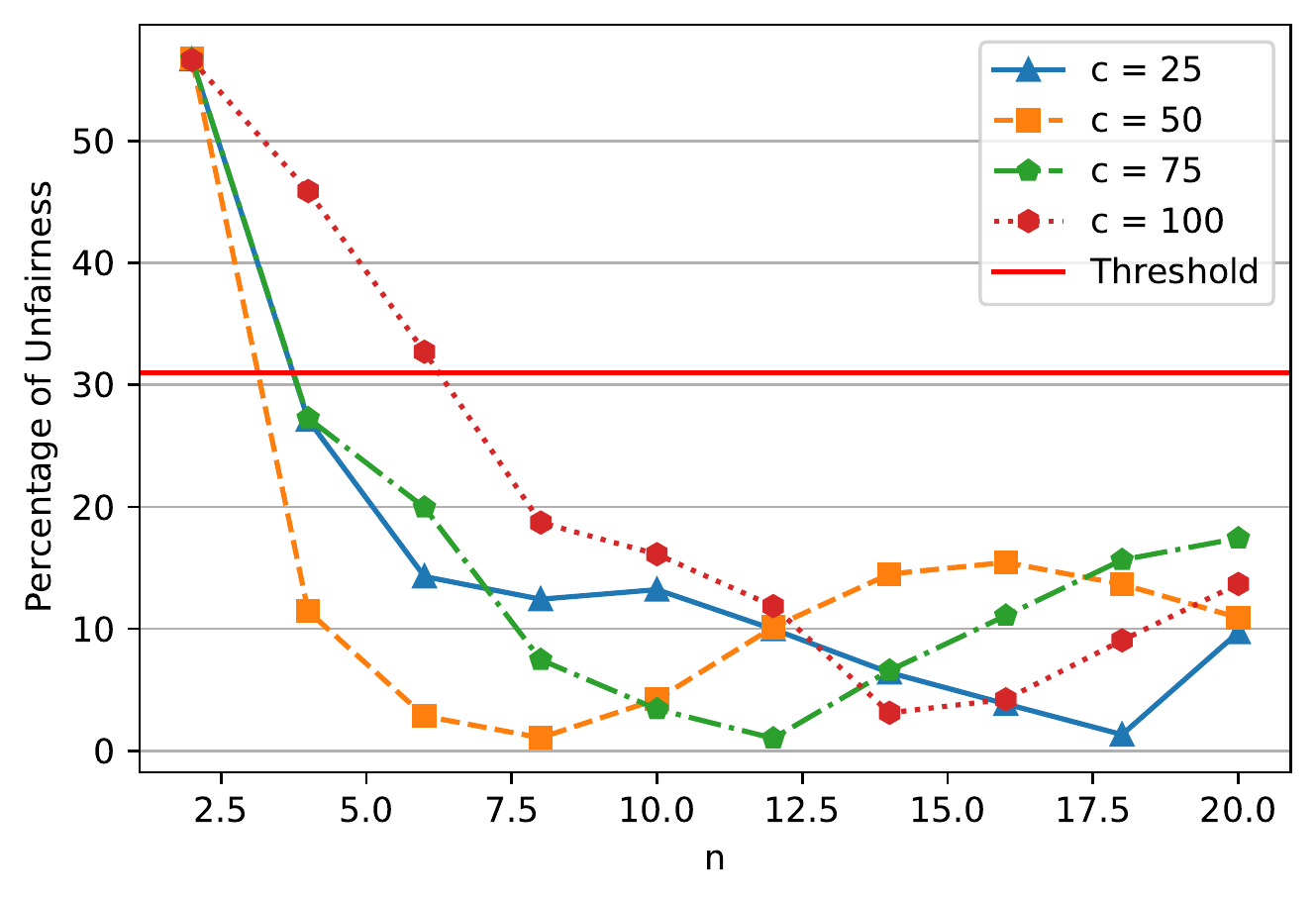}	
	}
	\vspace{-10pt}
	\hfill
	\subfloat[Fitting Error versus $c$.\label{fig: NewYork3}]{%
		\includegraphics[scale=.45]{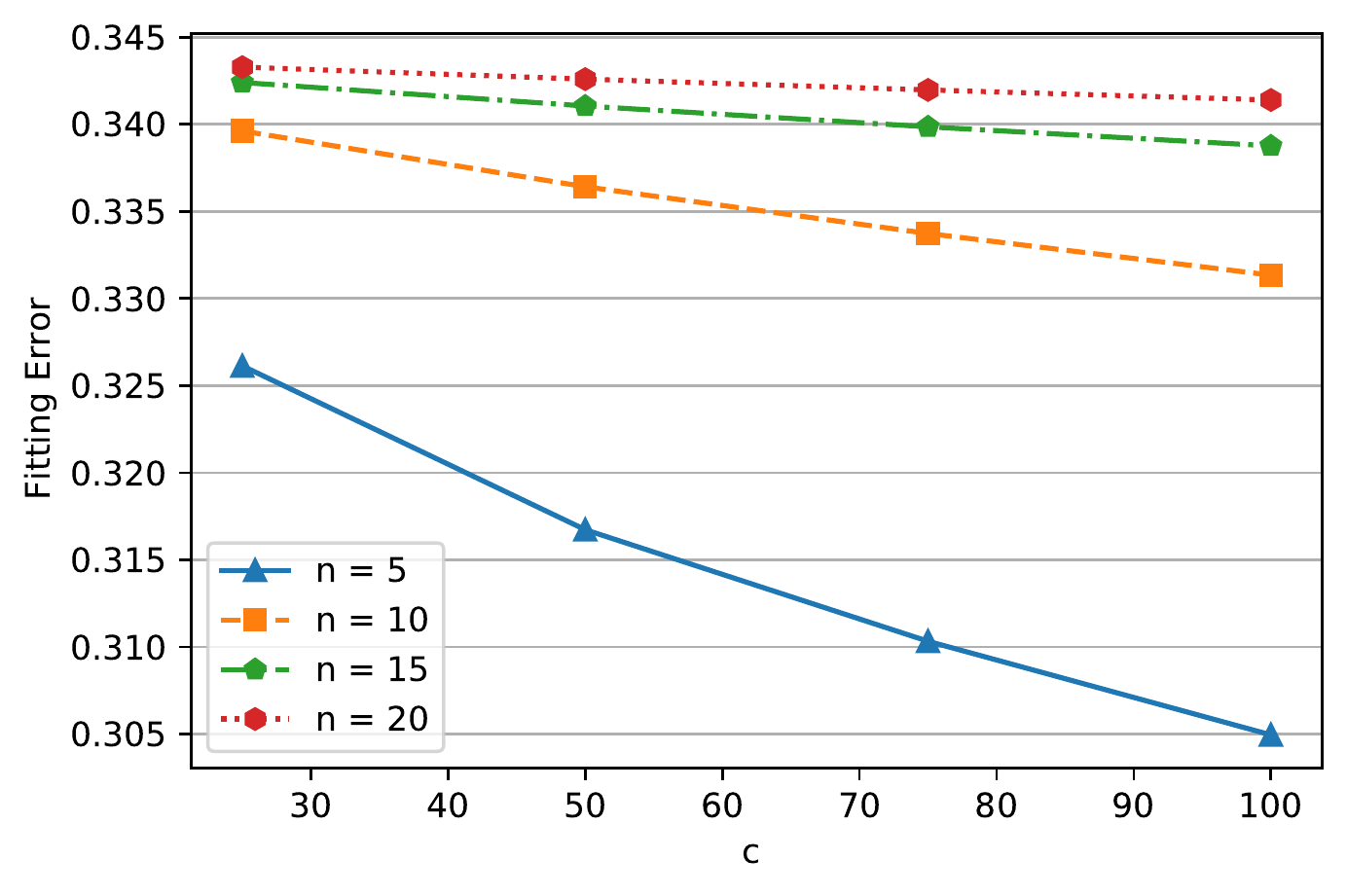}
	}
	\vspace{-10pt}
	\hfill
	\subfloat[Fitting Error versus $n$.\label{fig: NewYork4}]{%
		\includegraphics[scale=.45]{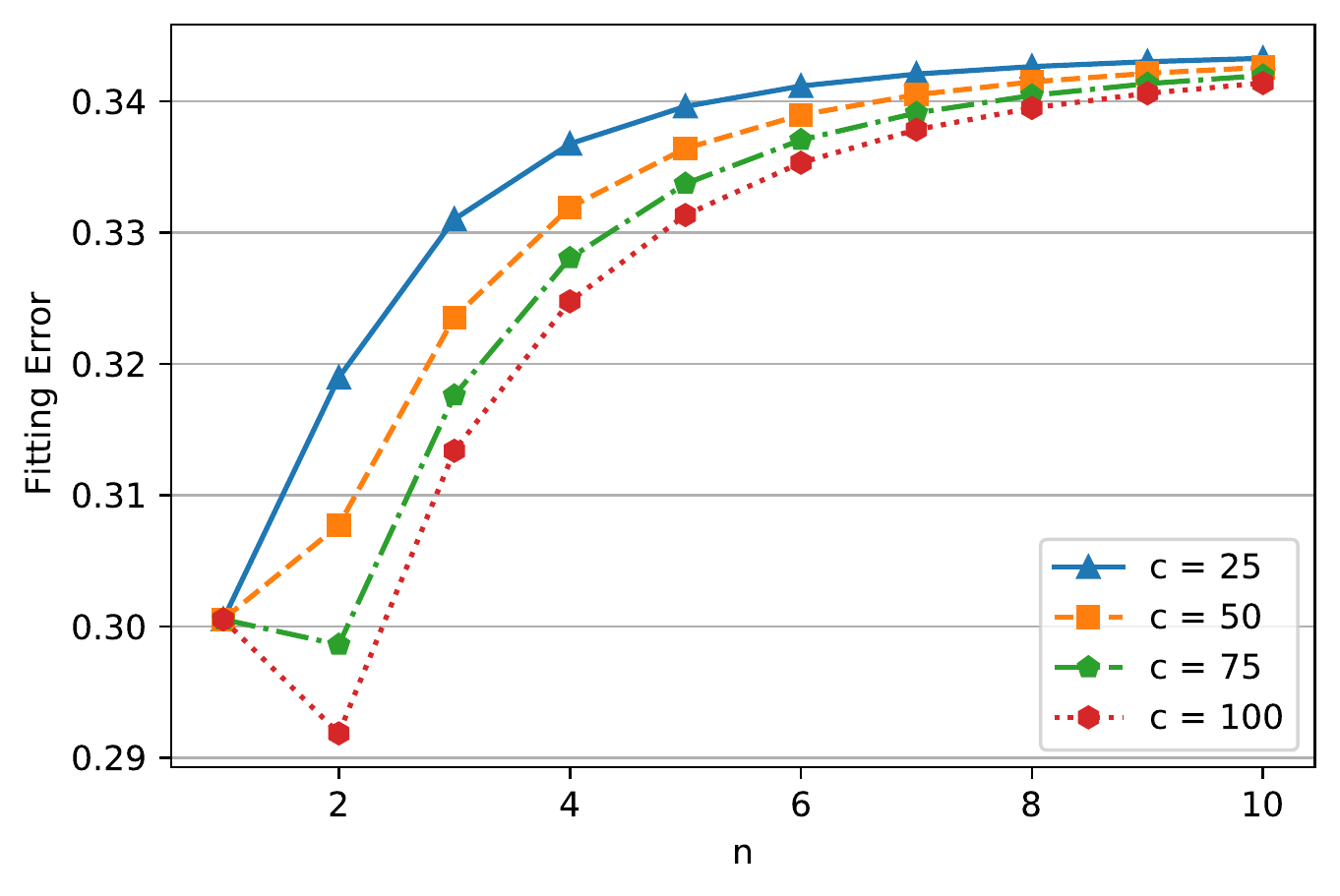}
	}
	\vspace{-10pt}
	\caption{Distance-based Mechanism, New York taxi dataset}
	\label{fig: average cost price}
	\vspace{-15pt}
\end{figure}

\begin{figure}[!t]
\centering
	\subfloat[Percentage of Unfairness versus $c$.\label{fig: Chicago1}]{%
	\includegraphics[scale = 0.45]{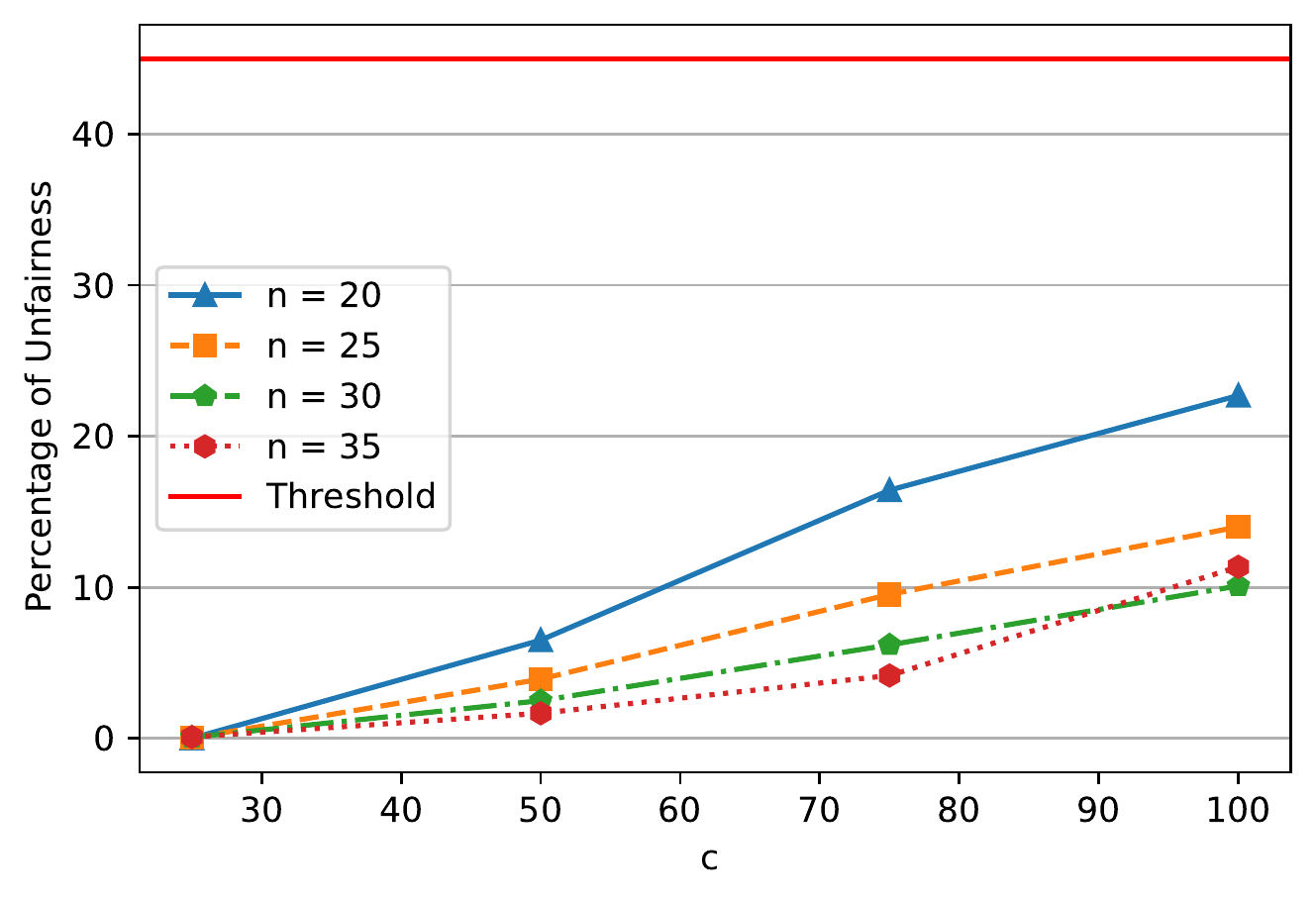}
	}
	\vspace{-10pt}
	\hfill
	\subfloat[Percentage of Unfairness versus $n$.\label{fig: Chicago2}]{%
		\includegraphics[scale = 0.45]{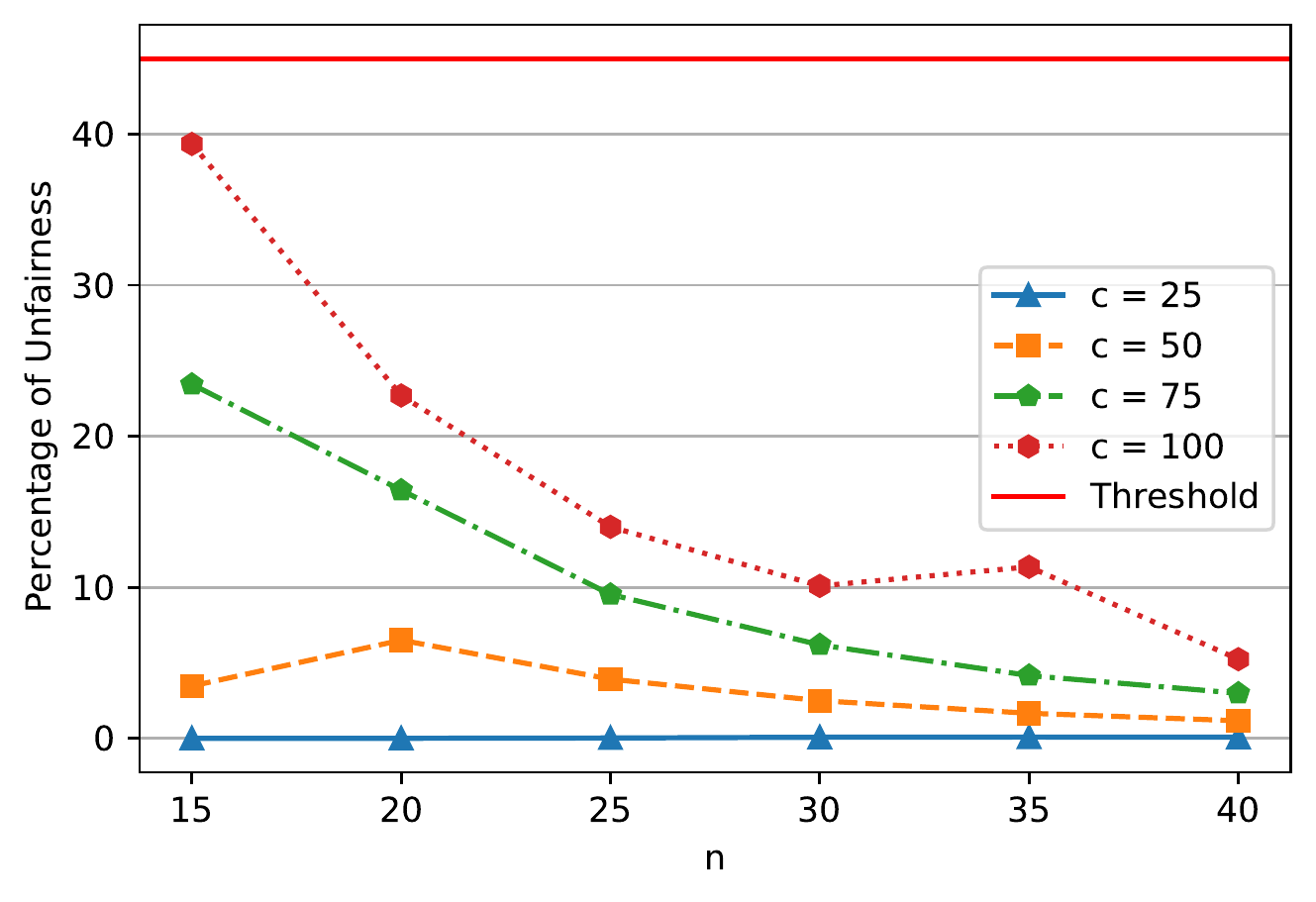}	
	}
	\vspace{-10pt}
	\hfill
	\subfloat[Fitting Error versus $c$.\label{fig: Chicago3}]{%
		\includegraphics[scale=.45]{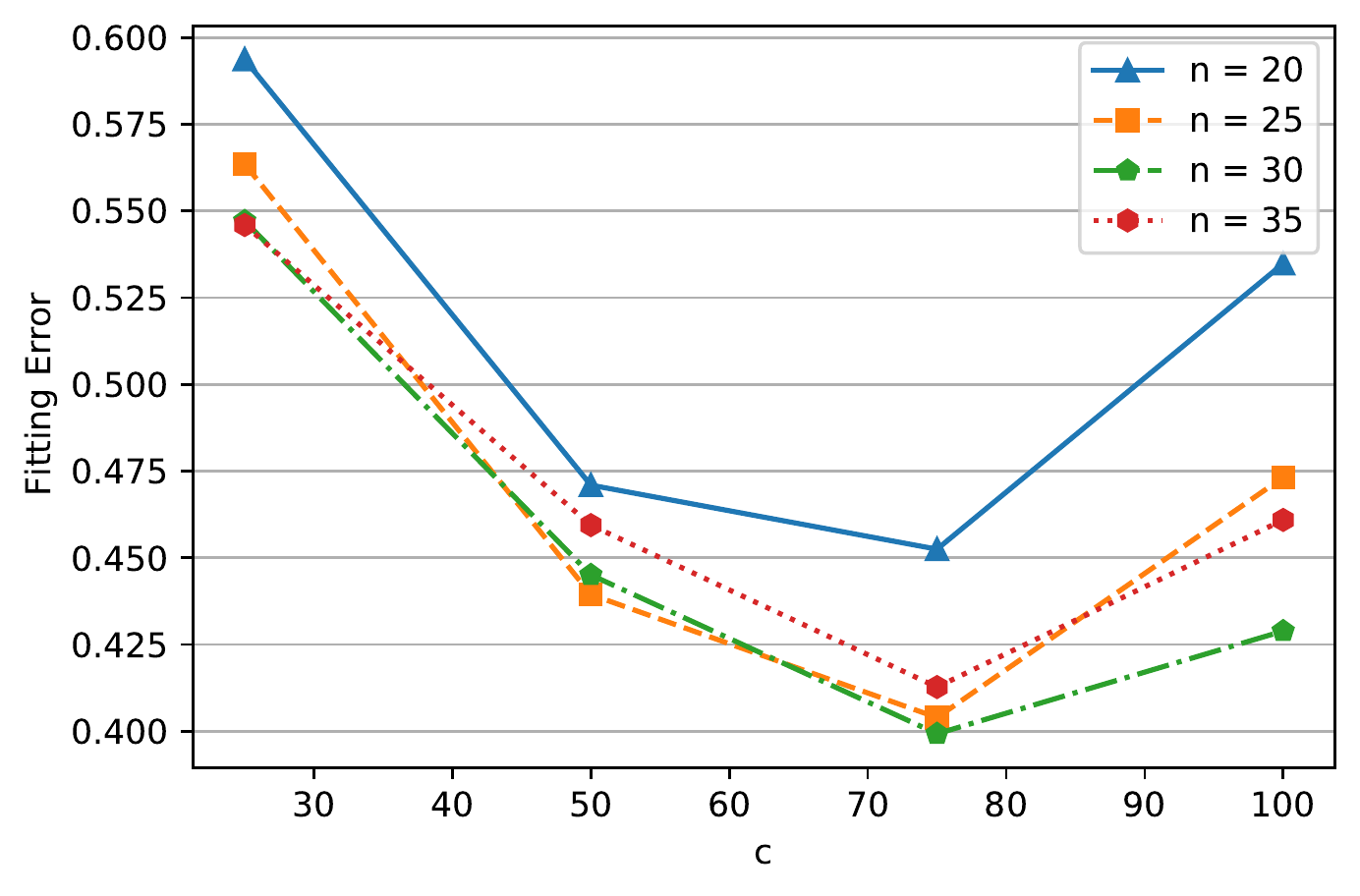}
	}
	\vspace{-10pt}
		\hfill
	\subfloat[Fitting Error versus $n$.\label{fig: Chicago4}]{%
		\includegraphics[scale=.45]{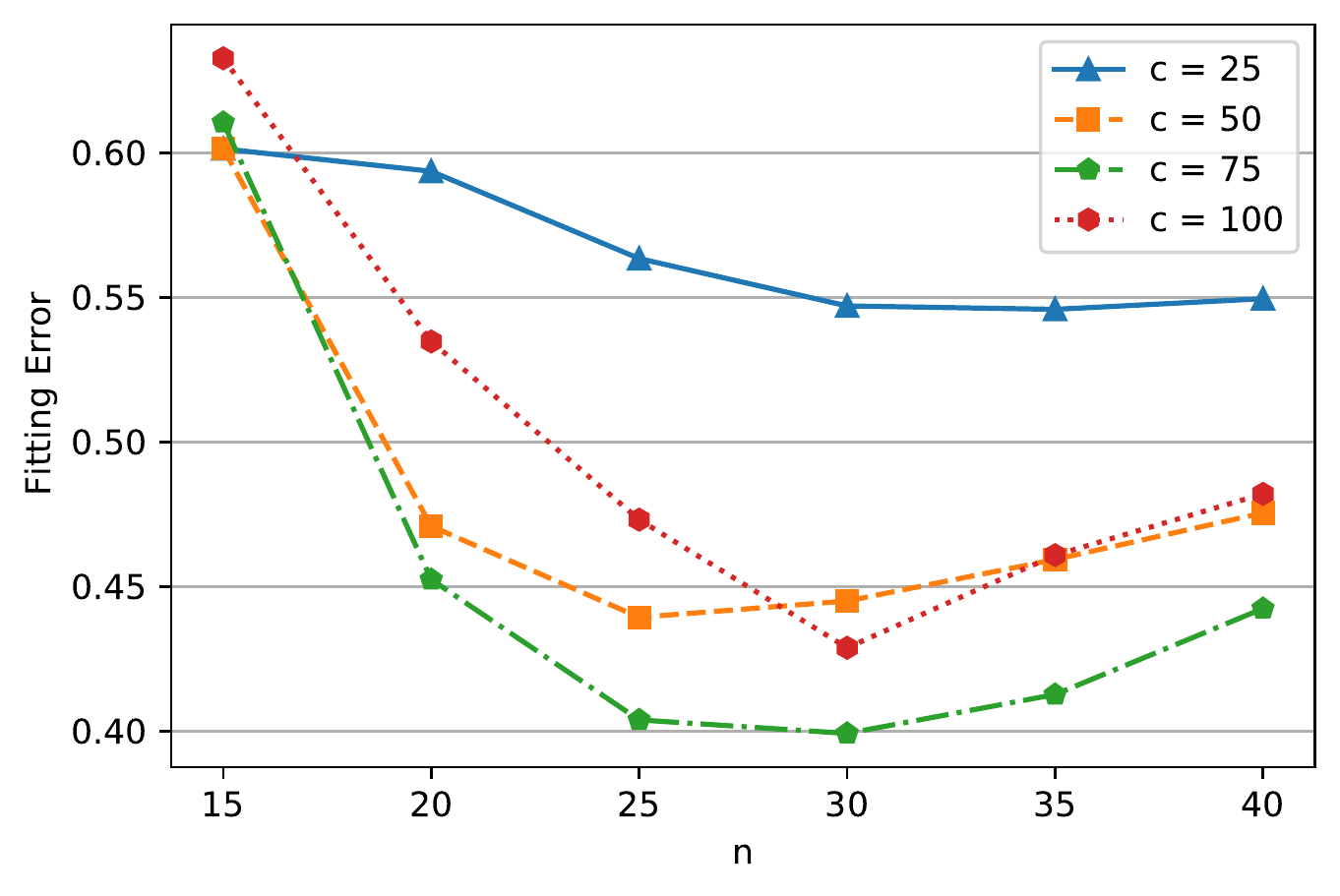}
	}
	\vspace{-10pt}
	\caption{Zone-based Mechanism}
	\label{fig: Results for Chicago}
	\vspace{-15pt}
\end{figure}

\section{Experimental Evaluation}\label{Sec: Experimental Evaluation}

We evaluate our proposed spatial data fairness mechanisms in the two studied scenarios. For the distance-based case, we use a dataset of taxi fares from New York City; for the zone-based case, we consider budget allocation to police departments according to the Chicago crime occurrence dataset. 
%The first set of experiments presented in Subsection~\ref{Sec: New York Taxi Fares} is focused on the application of $c$-fair polynomial for individual location fairness in DtR applied on top of a Deep Learning model, and the second set of experiments aims to evaluate the performance for coordinates, i.e., the generalized case. 
We ran experiments on a $3.40$GHz core-i7 Intel processor with 8GB RAM. The code is implemented in Python and uses the Trust Region Reflective least square implementation from Scipy~\cite{ScipyOptimizer} (the maximum number of iterations for convergence is set to $300$, and the default tolerance threshold value of $1e$-$2$ is used to stop the optimization iterations). %We expect some volatility in the results as the optimizer debuts the optimization by selecting a random point as a solution, and following several iterations to reach a near-optimal result.  

%Most existing fairness research focuses on group fairness, and it does so in the context of data types {\em other} than geospatial~\cite{fabris2022algorithmic}. Individual fairness enforcement requires the evaluation of $O(m^2)$ hard constraints, growing quadratically with the number of datapoints, which is more difficult to achieve efficiently. 
%There are no direct competitors to our approach to compare against as benchmarks, except for~\cite{riederer2017price}, which defines a specific type of objective function for location advertisements. Their approach is very slow, due to the fact that they do not rely on efficient constructions like our proposed fair polynomial approach. In their work, the authors of ~\cite{riederer2017price} evaluate their approach for two users only, and their technique cannot be deployed for large-scale datasets. Hence, we are not able to compare against that benchmark.

\subsection{Distance-based Spatial Fairness}\label{Sec: New York Taxi Fares}

We sampled 120,000 records from the NYC taxi dataset~\cite{kaggle} providing over 55 million trips and their associated fares. 
%The dataset is based on real-world data collected throughout New York City and used in a popular competition hosted by Google Cloud. 
We deployed an Artificial Neural Network (ANN) to assess the likelihood of taxi fares being fair in the system. \revision{Specifically, we seek to capture whether there is bias in the setting of fares based on the specific neighborhoods where the trip starts/ends. Ideally, the trip distance should be the only factor determining price (we carefully pre-process the dataset such that trips are clustered according to the time of week/day, such that differences in fare due to demand status and traffic causes are eliminated). Our goal is to first understand the percentage of records for which the individual fairness constraints do not hold with respect to traveled distances. Then, we analyze the performance of the proposed $c$-fair mechanism.}

{\em ML Model for Fairness Characterization.} Our ANN model consists of two hidden layers with $200$ and $100$ neurons and an output layer with two neurons representing the binary classification task. The activation function used in the model is RELU, the dropout probability for each layer set to 0.4, and cross-entropy is used as the loss function. The accuracy of the model is $92\%$. The input features include pick-up date and time (categorical hour, AM or PM, weekday, EDT date), pick-up longitude, pick-up latitude, drop-off longitude, drop-off latitude, passenger count, and distance traveled in kilometers. The ride fares have a mean of $10$ dollars with a standard deviation of $7$ dollars, and the average traveled distance is $3.31$ km with a standard deviation of $3.2$ km. For model training purposes, we split the data into training, validation and test datasets with $96,000$, $12,000$, and $12,000$ records, respectively. To generate the ground truth for the training dataset, we have used price per kilometer traveled as the indicator of how fair the associated traveled fares are. For every hour of the day, the average price per kilometer is calculated as the hard threshold between fair and unfair travels. %Fig.~\ref{fig: average cost price} represents average price per kilometer thresholds used for various hours of the day. 
The trips above the threshold are classified as unfair, and the values less than the threshold are assumed to be fair. This results in a total of $21,928$ trips being categorized as unfair.

Once the ANN model is trained, we predict the likelihood of each trip fare being fair on the test dataset. 
%Having generated the likelihood scores, they are considered with respect to the travel distances of taxis. 
For every two records, the individual fairness constraint is evaluated to reveal whether fares are fair with respect to travel distances. In the absence of any fairness mechanisms being deployed, $32\%$ of constraints are not satisfied, hence those trips are unfair. \revision {The $32\%$ threshold is highlighted with a red horizontal line in the experimental plots, to highlight the fairness improvement of the considered approaches.}

Next, we apply the proposed $c$-fair mechanism to achieve distance-based fairness. Our experiments evaluate the performance based on four key metrics: percentage of unfairness (constraints were not satisfied), the degree of $c$-fair polynomial ($n$), the parameter $c$, and the root mean square (RMS) of fitting error to likelihood scores.

{\em Percentage of Unfairness.} Fig.~\ref{fig: NewYork1} shows the impact of increasing $c$ on reducing unfairness when the degree of the polynomial is $5$, $10$, $15$, and $20$. As expected, lower values of $c$ result in higher fairness in the system, with maximum fairness achieved when $c$ is equal to one. For the maximum fairness scenario, the percentage of unfairness is zero, meaning that all individual fairness constraints are satisfied for every two records in the dataset. By increasing the value of $c$, fair polynomials would have more room for maneuver and fitting to likelihood scores, but it comes with the cost of higher unfairness. Such behaviour demonstrates the utility-fairness trade-off captured by the constant $c$. 
%The red horizontal line is shown as a baseline indicating the percentage of unfairness before applying the mechanism.  
Increasing the polynomial degree can be seen to improve the percentage of unfairness until it reaches the point where it overfits the likelihood scores, and the performance deteriorates. Fig.~\ref{fig: NewYork2} shows more clearly the impact of increasing the value of $n$ on unfairness. Lower $c$ values result in a lower percentage of unfairness for all polynomial degrees.

{\em Fitting Error.} Figs.~\ref{fig: NewYork3} and~\ref{fig: NewYork4} demonstrate the amount of utility loss in data due to fitting likelihood scores to a $c$-fair polynomial. Two key trends can be observed from the figures. First, increasing the value of $c$ lowers the fitting error. This is expected, as higher $c$ allows more flexibility for selecting coefficients and better fitting performance. Second, increasing the value of $n$ for the same value of $c$ raises the fitting error. To understand this behavior, one can intuitively look at the problem as allocating the same amount of budget among several buckets representing coefficients. Although increasing the degrees of freedom provides better fitting performance as higher degree monomials exist, it further restricts the budget for each coefficient. Thus, the lower degree monomials, which have a more significant impact on the performance, are allocated a lower amount of budget, negatively affecting the performance. 

{\em Computational Complexity.} We measure the computational overhead of $c$-fair polynomials in terms of time complexity, number of iterations before optimization convergence, and final optimization cost (\revision{the final optimization cost represents the value of the Scipy cost function upon reaching the solution~\cite{ScipyOptimizer}}). Fig.~\ref{Fig: NY Computation Complexity Analysis} shows the results. In each graph, the overhead is shown for four values of $c=25,\, 50,\, 75,\, 100$ plotted for varying polynomial degrees. Overall, the time complexity is in the order of milliseconds and does not limit the practical deployment of $c$-fair polynomials. The second graph illustrates the number of iterations before reaching the optimal point. The optimization process stops either by reaching the maximum number of iterations ($300$) or when the relative change in optimization cost remains below the tolerance threshold ($1e$-$2$). As explained previously, the slight oscillation in the performance is due to selecting a random start point for the optimization.  

Note that, increasing the degree of polynomial $n$ results in higher computational overhead.  This is expected, as more degrees of freedom (coefficients) lead to more effort for finding the optimal point. Another consistent behavior across all three figures is that increasing $c$ on average reduces the computation complexity cost and facilitates reaching the near-optimal point. The trend is more apparent in the final optimization cost figure, in which it can be clearly seen that a higher $c$ value leads to a lower cost.

\begin{comment}

\begin{figure}[t]
%\centering
\raggedright
\includegraphics[scale=.45]{Figures/Crime_distribution.pdf}
%\hspace{1em}
\centering
\vspace{-10pt}
\caption{Chicago Crime Dataset.}
\label{Fig: chicago distribution}
\vspace{-20pt}
\end{figure}

\end{comment}

\begin{figure*}[t]
	\subfloat[Time Complexity \label{NY Time}]{%
	\includegraphics[scale=.41]{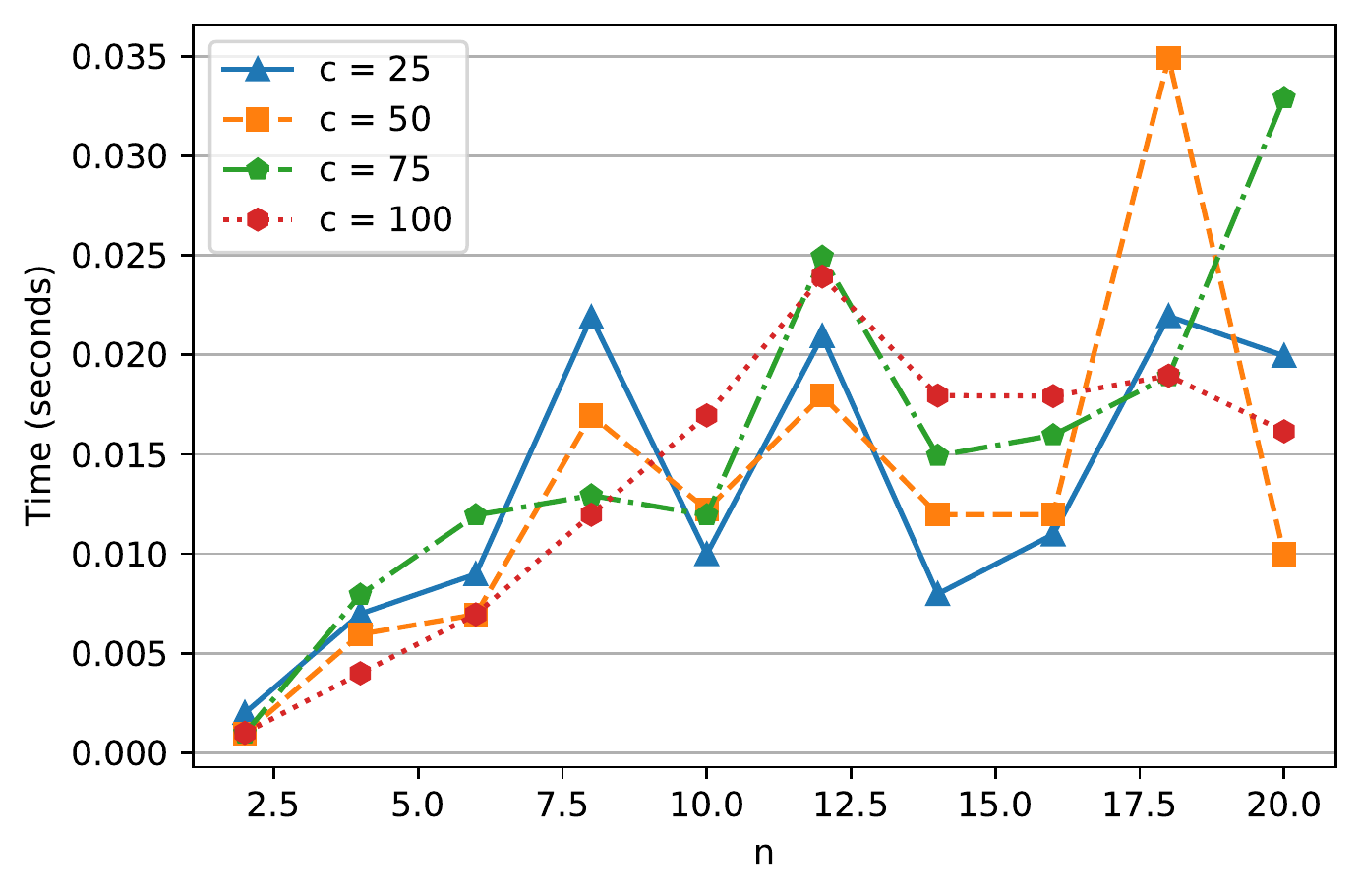}
	}
	\hfill
	\subfloat[Number of Iterations Before Convergence\label{NY Iterations}]{%
	\includegraphics[scale=.41]{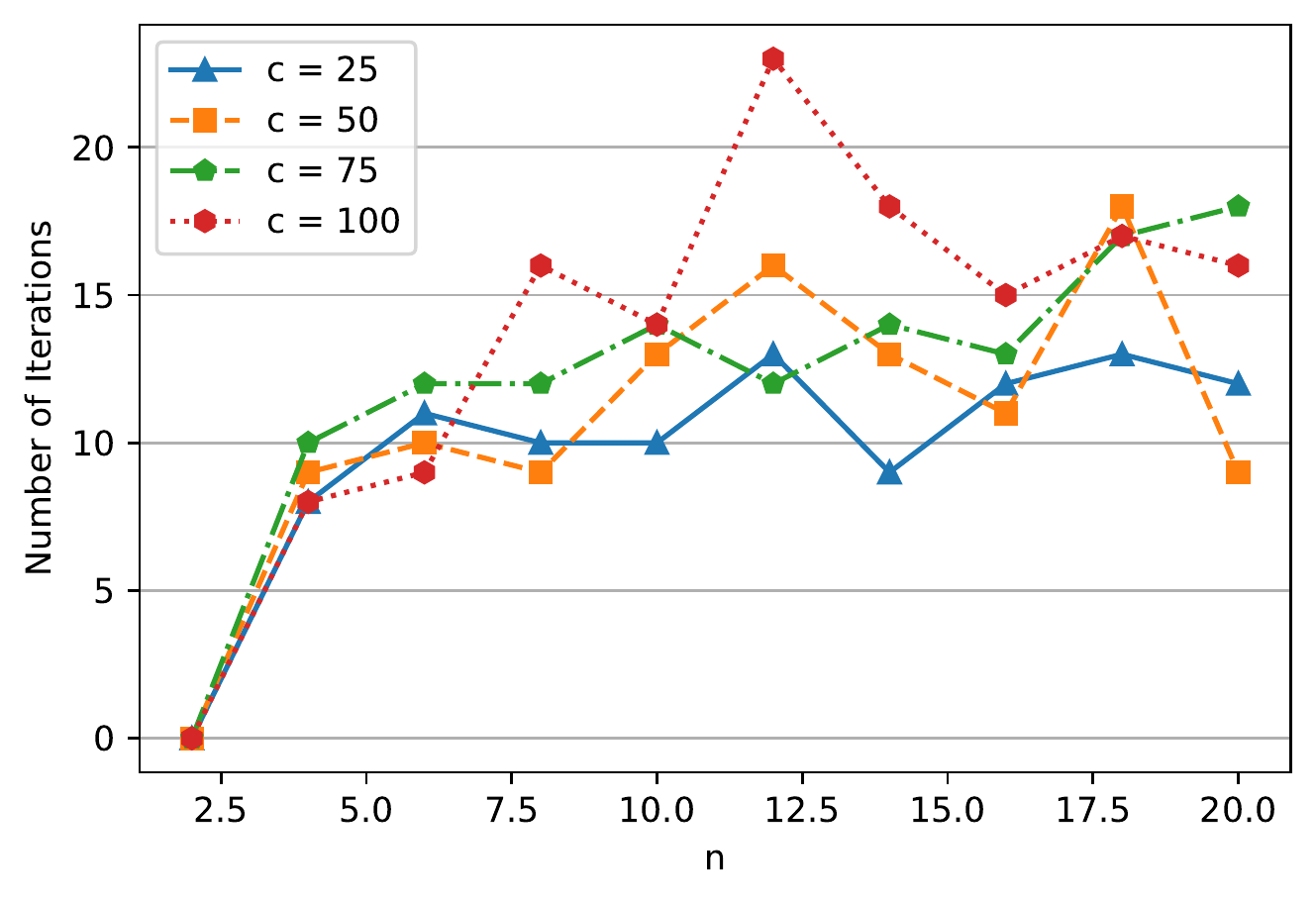}
	}
	\hfill
	\subfloat[Optimization Final Cost \label{NY Cost}]{%
	\includegraphics[scale=.41]{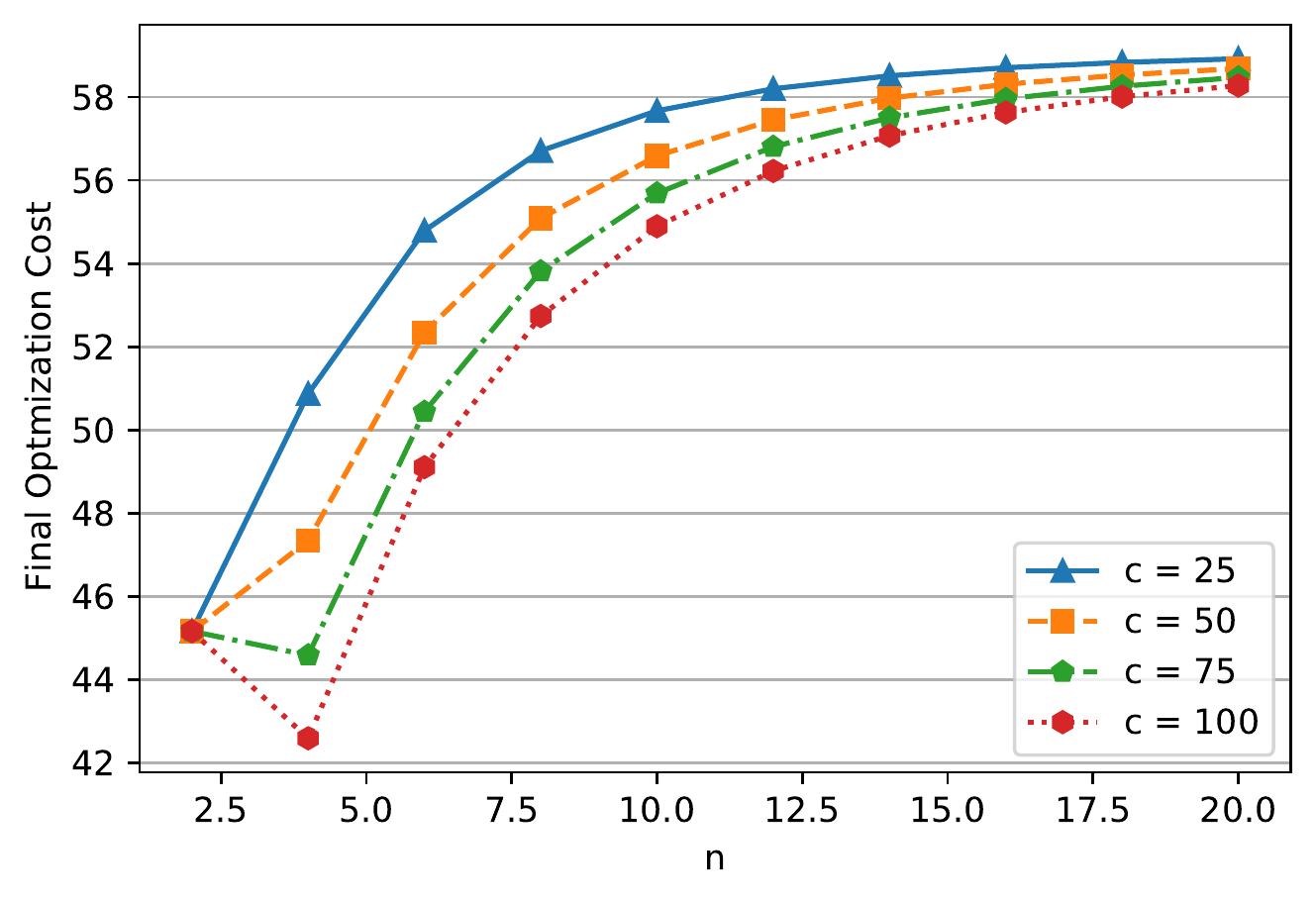}
	}
    \vspace{-10pt}	
	\caption{Computation Overhead Analysis on New York Taxi Dataset.}
	\label{Fig: NY Computation Complexity Analysis}
	\vspace{-10pt}
\end{figure*}

\begin{figure*}[t]
	\subfloat[Time Complexity\label{Chicago Time}]{%
	\includegraphics[scale=.41]{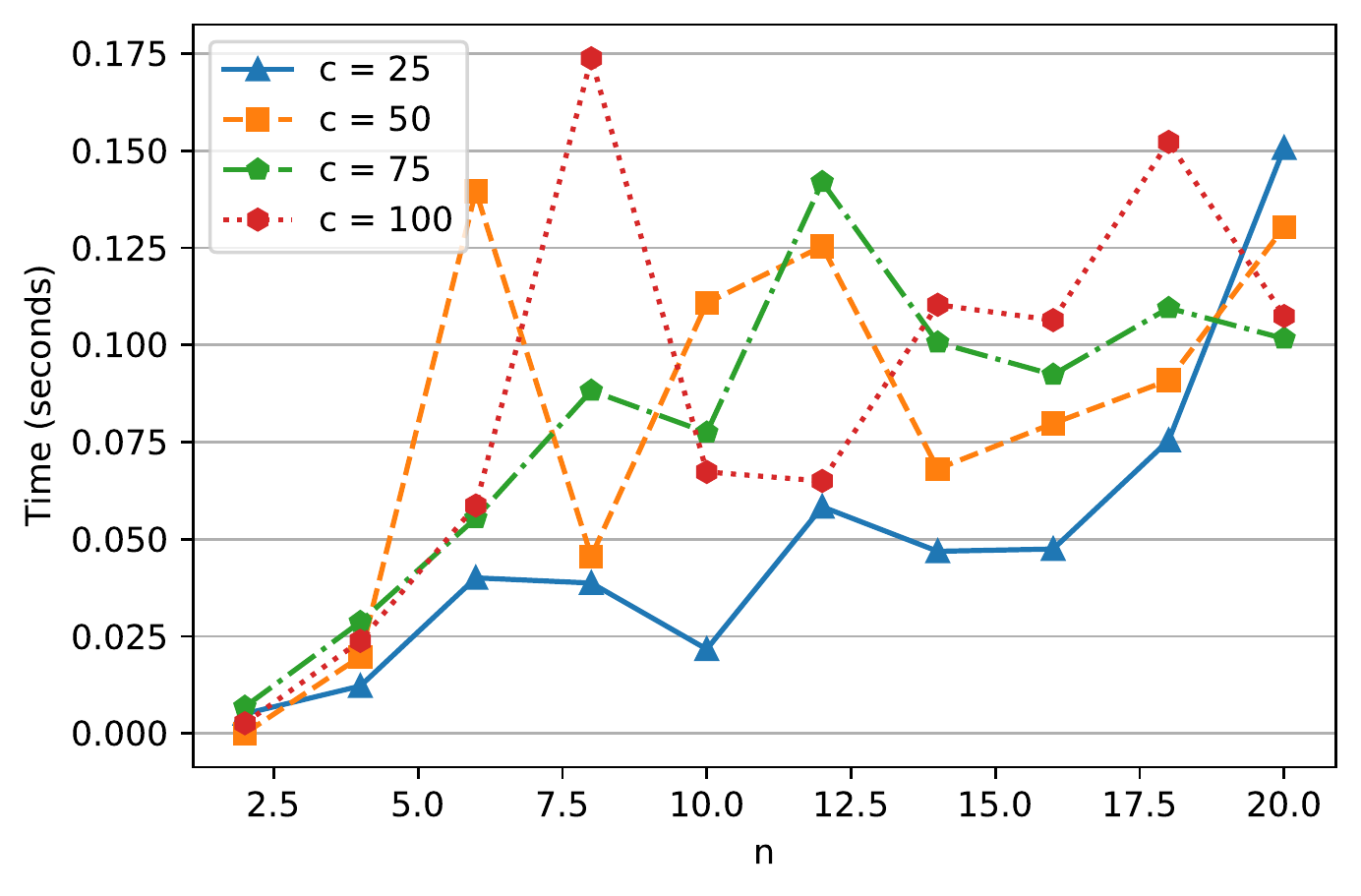}
	}
	\hfill
	\subfloat[Number of Iterations Before Convergence\label{Chicago Iterations}]{%
	\includegraphics[scale=.41]{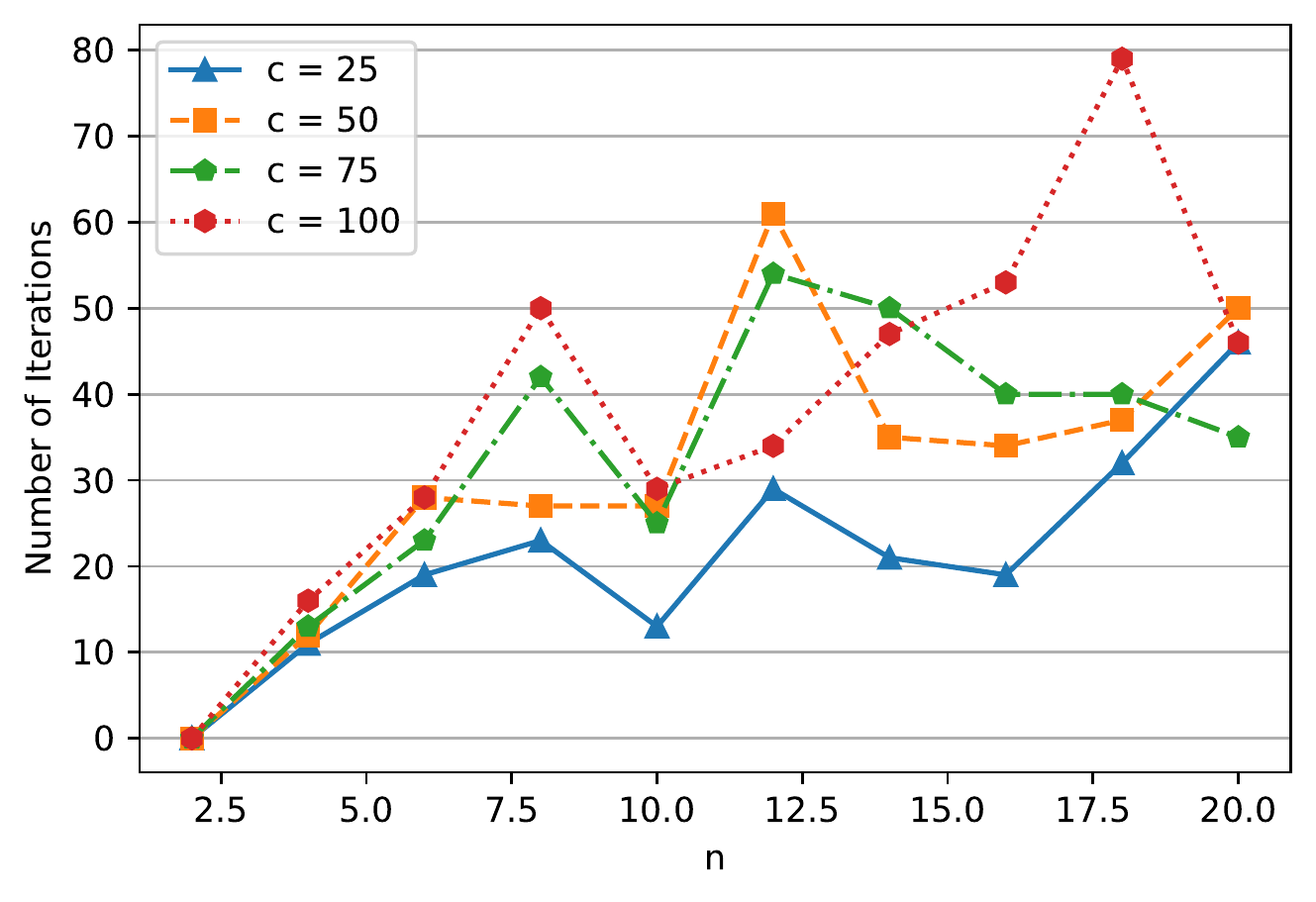}
	}
	\hfill
	\subfloat[Optimization Final Cost\label{Chicago Cost}]{%
	\includegraphics[scale=.41]{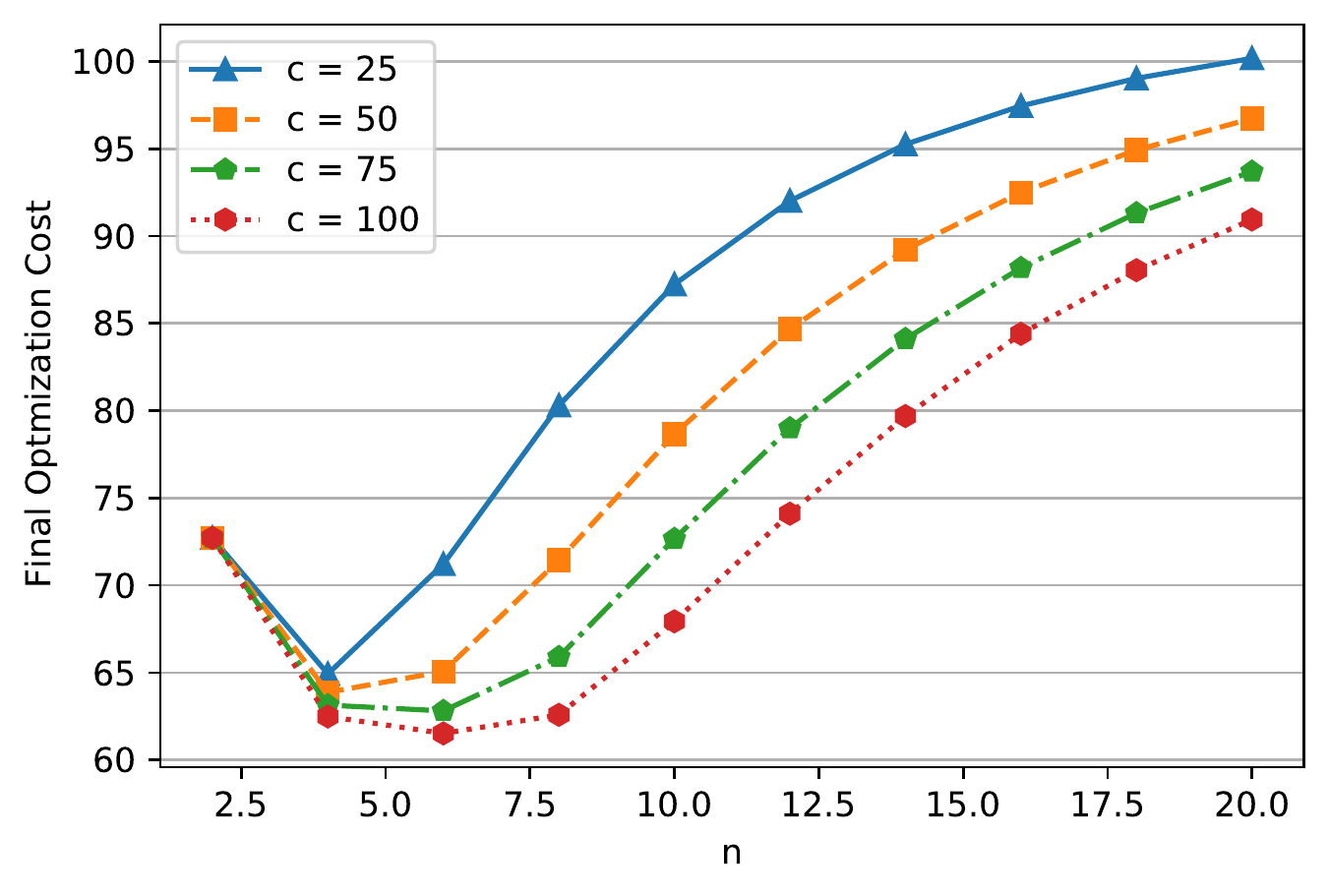}
	}
    \vspace{-10pt}	
	\caption{Computation Overhead Analysis on Chicago Crime Dataset.}
	\label{Fig: Chicago Dataset Computation Complexity Analysis}
	\vspace{-10pt}
\end{figure*}

\subsection{Zone-based Fairness}\label{Sec: Chicago Crime Zones}

For this scenario, we consider the case of budget allocation to different areas of Chicago, USA, based on the measured crime rates. We use the dataset provided by the Chicago Police Department's CLEAR (Citizen Law Enforcement Analysis and Reporting) system~\cite{chicago}, consisting of reported crime incidents in Chicago. A $1024\times 1024$ grid is overlaid on top of the Chicago map, and the goal is to fairly allocate the budget such that neighborhoods that are close to each other are treated similarly. We have selected seven major crime categories of sexual assault, homicide, kidnapping, sex offense, motor vehicle theft, criminal damage, and narcotics 
%(Fig.~\ref{Fig: chicago distribution}) 
among the reported crimes, and trained a logistic regression model to infer the likelihood of crime occurrence in each cell. The training dataset includes the crime data from January to November 2015, and the December data is chosen as the test dataset. The accuracy of the model is $94\%$ and its output is a set of likelihoods indicating the probability of crime occurrence. The budget allocated to each cell is proportional to the likelihood score derived by the classifier.

Once the likelihood scores are generated, they are used with $X$ and $Y$ cell coordinates to achieve individual location fairness with the distance metric set to $2$-norm. 
In absence of any fairness mechanism, we determine the percentage of individual location fairness constraints not being satisfied at $44.0 \%$. 

\revision{To understand if the expansion to higher polynomial degrees is essential, we started our experiments by focusing on degree one $c$-fair polynomials and applying the results in Theorem~\ref{Theorem: sufficient condition2}. As expected, the fitting error was rather high, and the utility was insufficient. We also noticed that for degree one polynomials, the optimal solution is achieved even when the value of $c$ is equal to one. Therefore, increasing $c$ does not help with improving the fitting error. Thus, it is crucial to use higher degree polynomials for this purpose.}

\begin{figure}[t]
	\subfloat[Distance-based Mechanism\label{r_1}]{%
	\includegraphics[scale=.3]{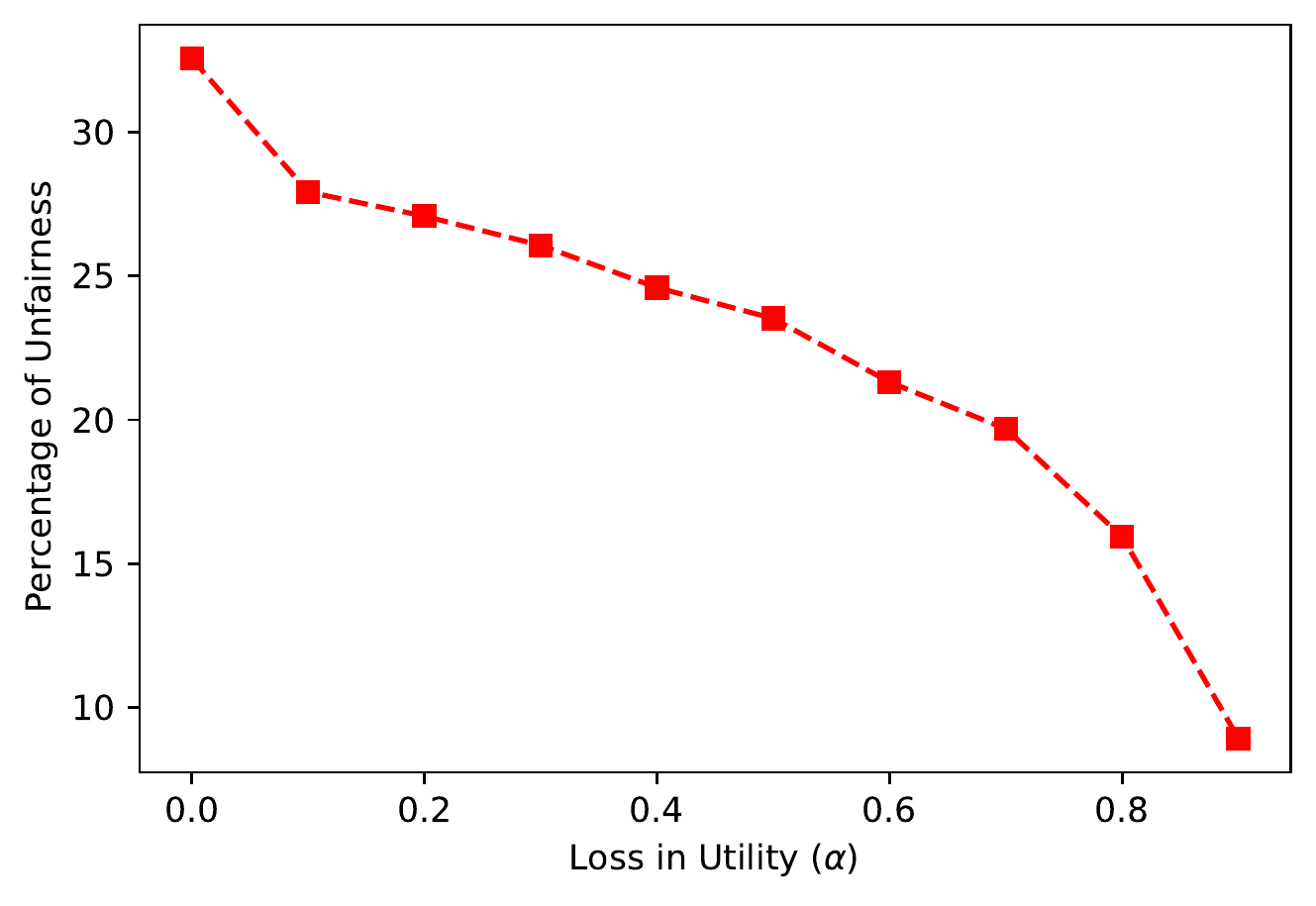}
	}
	\hfill
	\subfloat[Zone-based Mechanism\label{r_2}]{%
	\includegraphics[scale=.3]{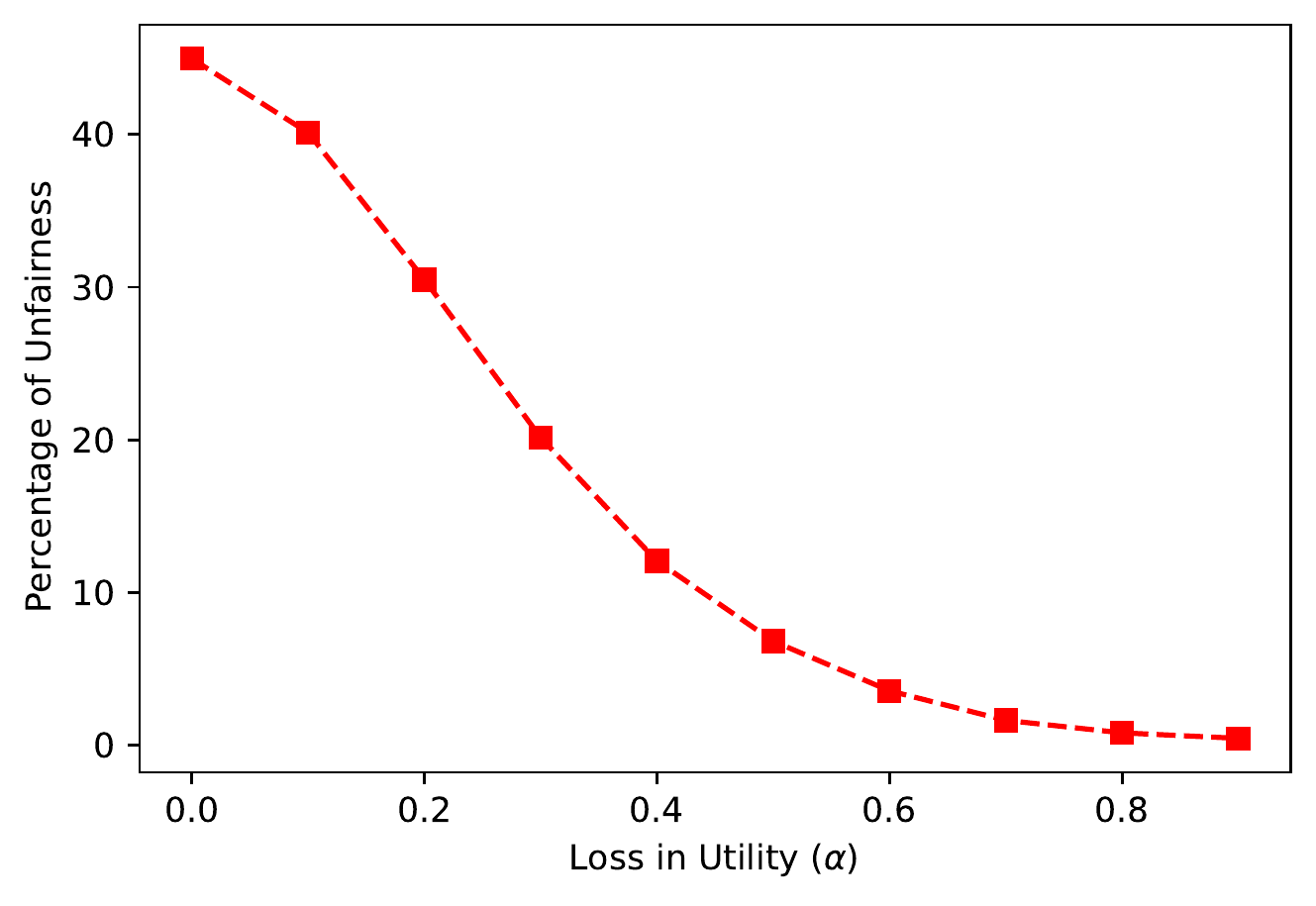}
	}
	\vspace{-10pt}
	\caption{Comparison Benchmark Performance.}
	\label{Fig: Baselines}
	\vspace{-15pt}
\end{figure}	

Next, we apply the optimization formulation derived in Theorem~\ref{Theorem: sufficient condition4}, the most generalized formula allowing each dimension to contribute in fitting with a degree $n$ polynomial. Fig.~\ref{fig: Results for Chicago} demonstrates the performance of $c$-fair polynomials for achieving individual location fairness on the crime dataset. The patterns are generally consistent with the distance-based case considered in the New York taxi fares experiment. Figs.~\ref{fig: Chicago1} and~\ref{fig: Chicago2} show the impact of $c$ and $n$ on the percentage of unfairness and Figs.~\ref{fig: Chicago3}  and~\ref{fig: Chicago4} show the performance with respect to utility. The red line is used as the reference point representing the percentage of unfairness in the original data, in the absence of fairness mechanisms.

Increasing the value of $c$  results in a lower degree of fairness and higher fitting error once the degree of polynomial reaches an acceptable level. This result further substantiates the fairness-utility trade-off in the system, also observed in the distance-based case. In a 2-dimensional space, using a degree of $10$ and the above polynomial to model each dimension can ensure that scores are under-fitted, preventing high fitting error values. In summary, the amount of fairness achieved with the fair polynomials, even for a reasonably low degree for polynomials such as $15$ and values of $c$ greater than $10$, can be seen to be over $70\%$. 

Fig.~\ref{Fig: Chicago Dataset Computation Complexity Analysis} shows the computation complexity of the proposed mechanism. The first point to notice is that a relatively high amount of time is required to achieve zone fairness  compared to the distance-based setting. However, computational complexity is still low in absolute value, and not an obstacle for practical deployment, with sub-second execution time. 
%The overhead evaluation on the Chicago crime dataset confirms and follows a similar trend as the computational complexity performance on the New York taxi dataset with a negligible difference in volatility. The rising volatility is caused in having higher dimensions and more complex optimization problem. 

\subsection{Comparison with benchmarks}
\label{Section: bench}
\revision{
We derive a threshold-based benchmark from the fairness mechanism in~\cite{dwork2012fairness} and the binary evaluation approach in~\cite{fabris2022algorithmic}. 
%Consider a dataset $\mathcal{L}$ including location data, and let us denote the corresponding DtRs of data entries by $\{l_1,..., l_m\}$. Moreover, let the function $M: \mathcal{L}\rightarrow [0,1]$ return the generated ML scores. 
Given constant threshold $t$, let polynomial $P(l_i)=t$, and let parameter $\alpha$ define by how much scores can be altered (e.g., if $\alpha=0.1$, each score can be altered by at most $\pm 0.1$). For DtR, the benchmark cycles through each score $M_i$ and pushes it towards polynomial $P(l_i)$ with an allowance of $\alpha$:
%\vspace{-10pt}
\begin{equation}
    M(l_i) \leftarrow M(l_i) + sign(P(l_i)-\alpha)\times  \min(\alpha, |P(l_i)-\alpha|)
\end{equation}
\vspace{-10pt}
}

\begin{figure*}[t]
	\subfloat[$c=25$ \label{r_1_1}]{%
	\includegraphics[scale=.16]{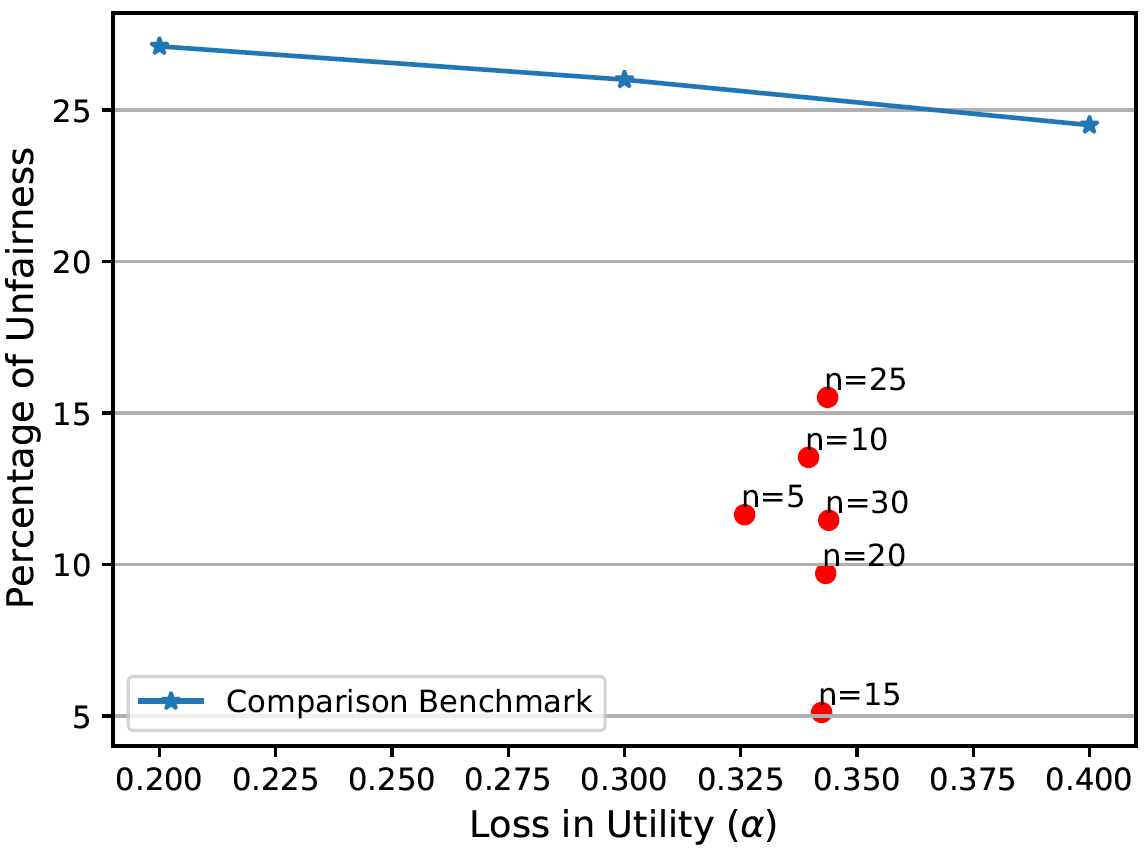}
	}
	\hfill
	\subfloat[$c=50$\label{r_1_2}]{%
	\includegraphics[scale=.16]{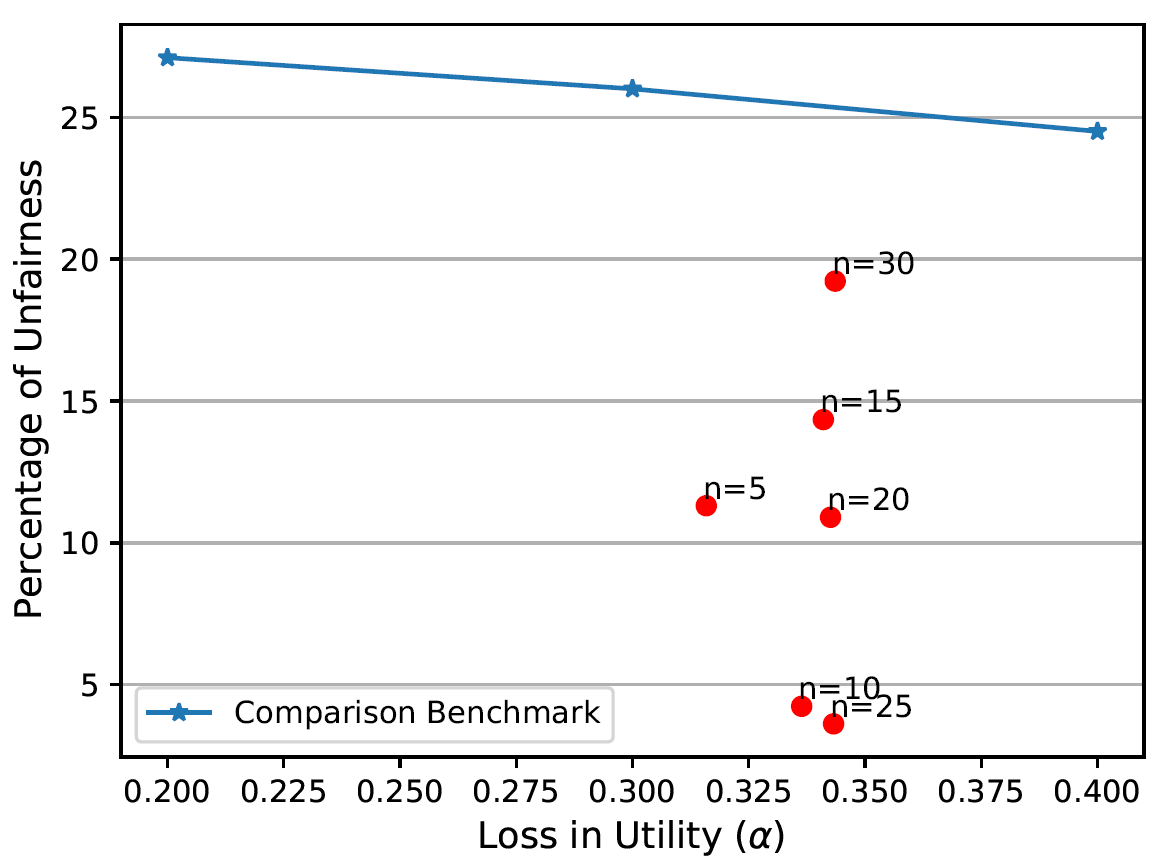}
	}
	\hfill
	\subfloat[$c=75$ \label{r_1_3}]{%
	\includegraphics[scale=.16]{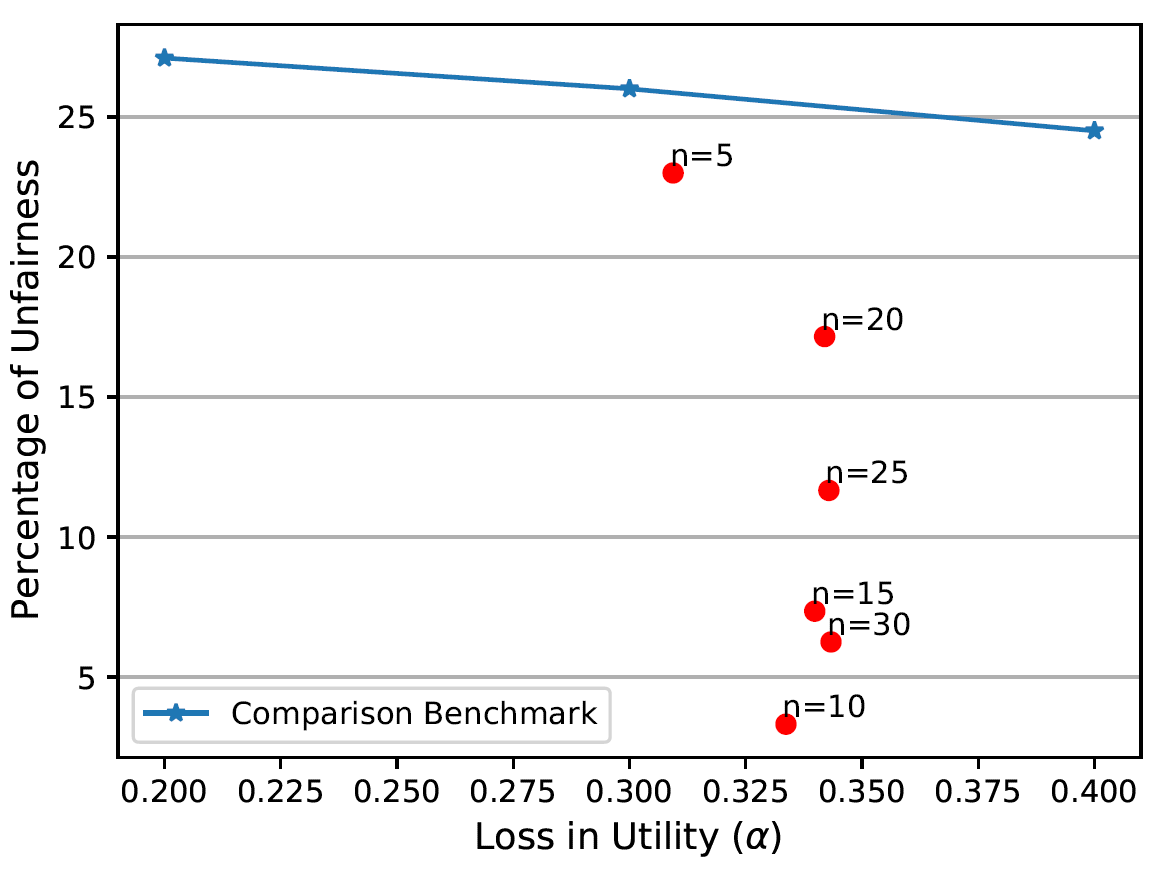}
	}
	\vspace{-10pt}
	\caption{\revision{Distance-based Mechanism Benchmark Comparison, New York taxi dataset}}
	\label{Fig: NY DtR Benchmark Comparison}
	\vspace{-10pt}
\end{figure*}

\begin{figure*}[t]
	\subfloat[$c=25$\label{r_2_1}]{%
	\includegraphics[scale=.16]{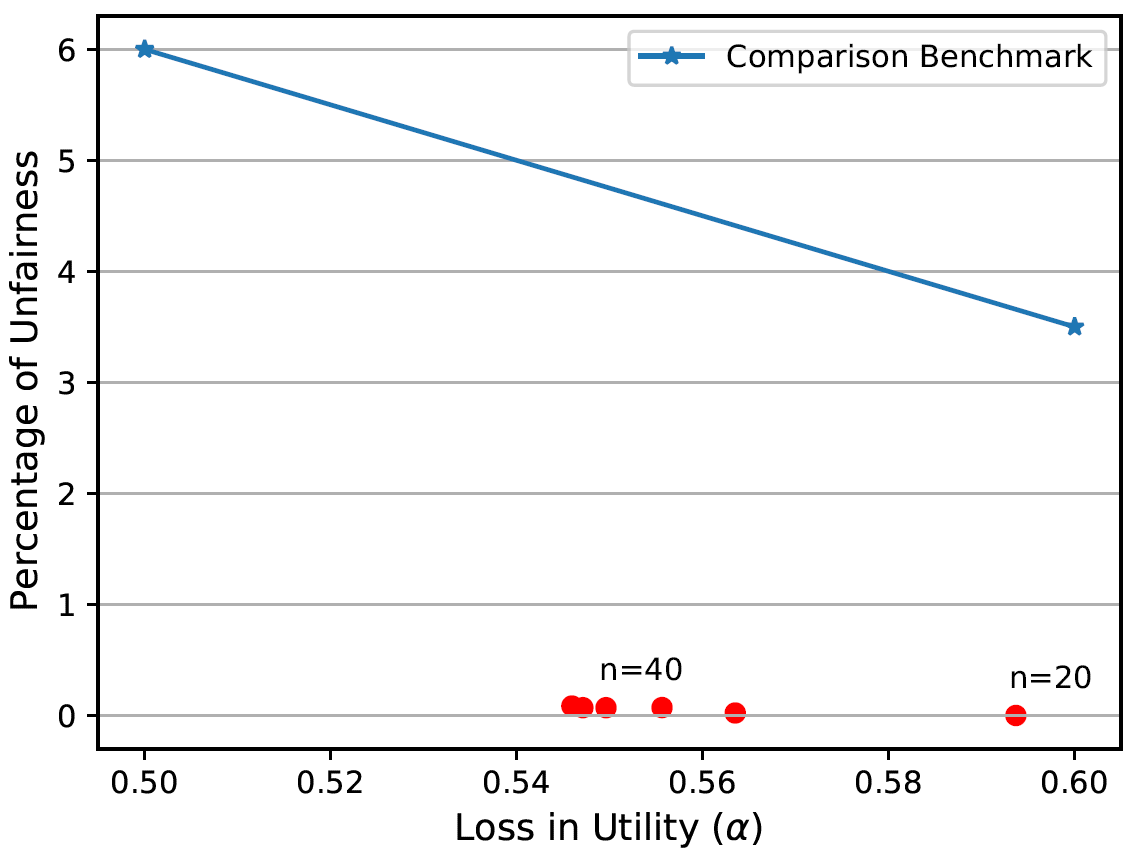}
	}
	\hfill
	\subfloat[$c=50$\label{r_2_2}]{%
	\includegraphics[scale=.16]{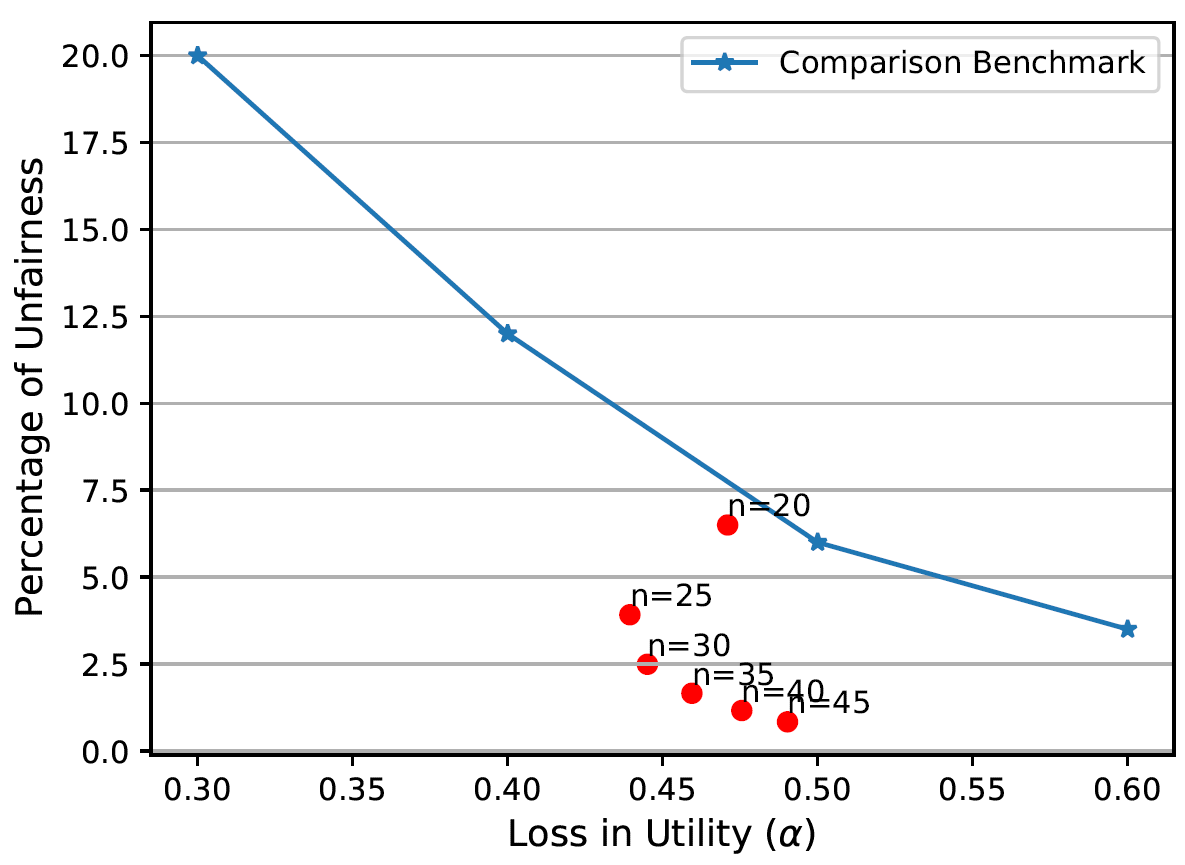}
	}
	\hfill
	\subfloat[$c=75$\label{r_2_3}]{%
	\includegraphics[scale=.16]{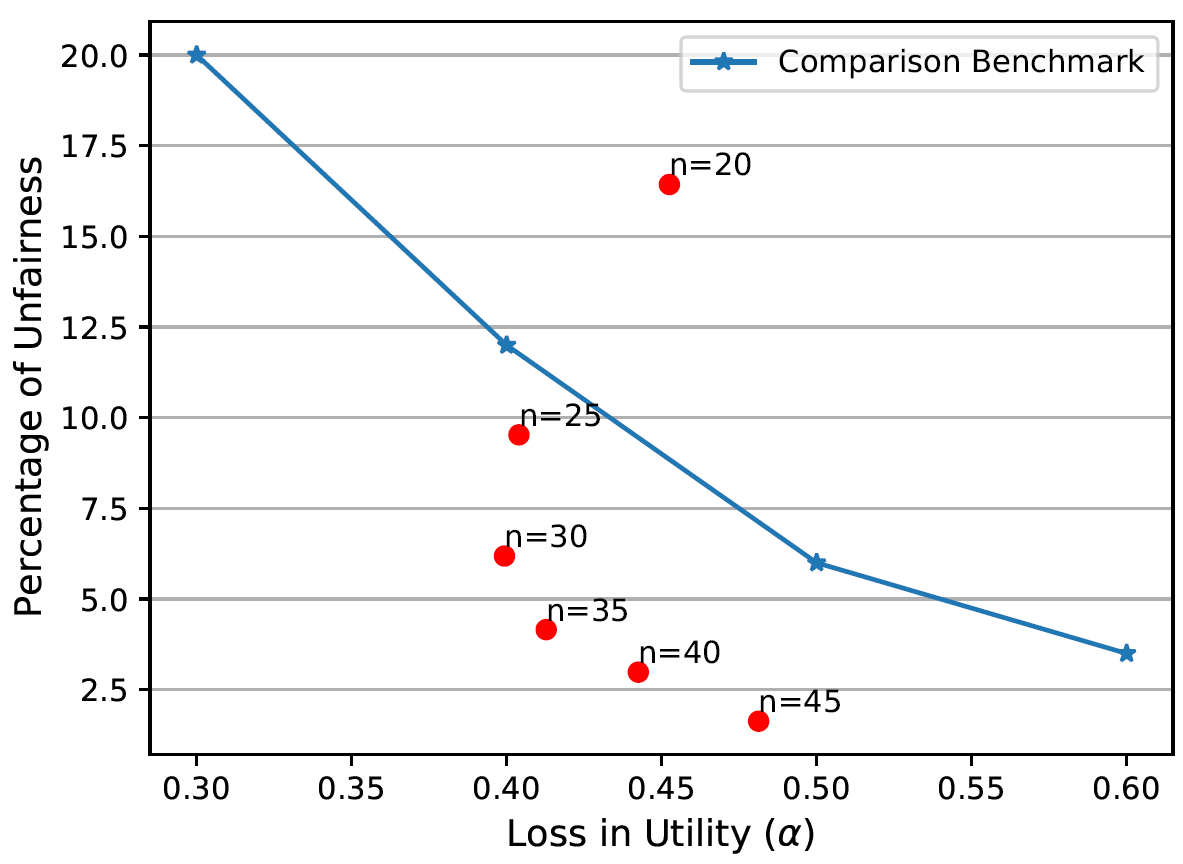}
	}
    \vspace{-10pt}
	\caption{\revision{Zone-based Mechanism Benchmark Comparison, Chicago crime dataset}}
	\label{Fig: Chicago DtR Benchmark Comparison}
	\vspace{-10pt}
\end{figure*}

\revision{The $sign$ operation ensures the direction of change is favorable to the benchmark, and $min$ ensures that the change in score does not overshoot $P(l_i)$. When $\alpha$ is zero, no utility loss exists, and the unfairness percentage is the same as the original. As $\alpha$ grows, the flexibility margin results in more constraints being satisfied until absolute fairness is achieved. 
For the zone-based fairness case, we use $P(x,y) = (x+y)/\sqrt{2}$, easily extensible to higher dimensions.}
%It is important to note that the polynomial has to be fair, meaning that it results in absolute fairness. A simplistic approach to select $P$ would be setting its value to a constant; however, in doing so, not only is the utility lost but the ranking between the scores is also completely lost. Therefore, we use the following polynomials in $1D$ and $2D$. }

\begin{comment}
\begin{itemize}
    \item \revision{In the DtR problem, which represents the $1D$ scenario, the most natural way to select the polynomial is by setting it to $P(l_i)=l_i$. It is easy to see why this polynomial results in absolute fairness,}
    \begin{equation}
         |P(l_i)- P(l_j)| \leq |l_i-l_j|, \;\; \forall \, i,\, j   
    \end{equation}
    \item \revision{In 2D, which considers individual fairness w.r.t. coordinates, and given Euclidean distance as the proximity metric, the polynomial $P(x,y) = (x+y)/\sqrt{2}$ is one of the primary choices as it results in absolute fairness proven below. }

\revision{For every two locations $\boldsymbol{l_1} = (x_1,x_2)$ and $\boldsymbol{l_2} = (x_1',x_2')$, we must have,}

\begin{equation}
|P(x_1,x_2)- P(x_1',x_2')| \leq  \sqrt[2]{(x_1-x_1')^2+ (x_2-x_2')^2}
\end{equation}

\revision{By setting  $P(x,y) = \dfrac{x+y}{\sqrt{2}}$ and applying Cauchy Schwarz inequality, it can be shown that the above inequality withholds for every two locations. }
\end{itemize}

\end{comment}

\revision{Figure~\ref{Fig: Baselines} shows the utility-fairness trade-off obtained by the baseline for both distance- and zone-based cases. Unfairness is monotonically decreasing as one allows a higher $\alpha$ alteration to the score. We notice a plateau of desirable behavior concentrated between $\alpha=[0.2,0.6]$, where the utility-fairness trade-off is good (for $\alpha$ values below 0.2 unfairness is too high, whereas for values higher than 0.6 the fairness gain fades in the multi-variate case).}

\revision{
Next, we focus our attention to this desirable $\alpha$ range, and we plot the relative performance of the benchmark versus our proposed fair polynomials approach for various values of $c$ and $n$. 
%Therefore, we zoom in on the graphs and compare them to fair polynomials in Figs.~\ref{Fig: NY DtR Benchmark Comparison} and~\ref{Fig: Chicago DtR Benchmark Comparison}. 
The plots in Figs.~\ref{Fig: NY DtR Benchmark Comparison} and~\ref{Fig: Chicago DtR Benchmark Comparison} show that our approach provides a superior trade-off compared to benchmarks. The benchmark outperforms only for zone-based fairness when $n=20$. In all other cases, fair polynomials provide either a vastly improved utility, or better fairness.}

\vspace{-5pt}

\section{Conclusion}
\label{Sec: Conclusion}

We studied in depth the problem of individual fairness for location data, and we identified sources of location bias that can occur in practical settings. We formulated two distinct problems that are relevant to location-based applications, and we devised specific techniques to achieve fairness while preserving utility, with the help of a novel construction called fair polynomials. While our focus is on geospatial data and applications, fair polynomials have the potential to provide useful building blocks for fairness in other application domains. In future work, we plan to study more complex types of location-based interaction. At the same time, we plan to study the effect of fairness mechanisms in conjunction with other constraints, such as privacy, e.g., devise mechanisms that can achieve both fairness and privacy.

\begin{acks}
This research has been funded in part by NSF grants IIS-1910950,
IIS-1909806, CNS-2027794, CNS-2125530, NIH grant R01LM014026, the USC Integrated Media Systems Center
(IMSC), and unrestricted cash gift from Microsoft Research and Google. Any opinions,
findings, and conclusions or recommendations expressed in this material are
those of the author(s) and do not necessarily reflect the views of any
of the sponsors such as the NSF.
\end{acks}

%\newpage

\bibliographystyle{ACM-Reference-Format}
\bibliography{sample}

\balance

\newpage

%\newcolumn

\appendix
\section{Appendix (Proof of Lemmas)}\label{Section: appendix}

\begin{lemma}[\em Generalized Titu's Lemma]
Let $m$ be an integer greater than or equal to $2$, $a_i^m$ a non-negative real number, and $x_i$ a positive real number. Then,
\begin{equation}
    n^{m-2} \sum_{i\Equal 1}^{n} \dfrac{a_i^m}{x_i}\geq \dfrac{(\sum_{i\Equal1}^n a_i)^m}{\sum_{i\Equal1}^n x_i}
\end{equation}
\end{lemma}

\begin{proof}
By Holder's inequality 
\begin{align}
(\sum_{i\Equal 1}^{n} 1)^{\dfrac{m-2}{m}}(\sum_{i\Equal 1}^{n} \dfrac{a_i^m}{x_i})^{\dfrac{1}{m}}(\sum_{i\Equal 1}^{n} x_i)^{\dfrac{1}{m}}&\geq \\
\sum_{i\Equal 1}^{n}& 1^{\dfrac{m-2}{m}}( \dfrac{a_i^m}{x_i})^{\dfrac{1}{m}}x_i^{\dfrac{1}{m}}\\
\rightarrow n^{\dfrac{m-2}{m}}(\sum_{i\Equal 1}^{n} \dfrac{a_i^m}{x_i})^{\dfrac{1}{m}}(\sum_{i\Equal 1}^{n} x_i)^{\dfrac{1}{m}}\geq  &\sum_{i\Equal 1}^{n} a_i\\
\rightarrow     n^{m-2} \sum_{i\Equal 1}^{n} \dfrac{a_i^m}{x_i}\geq \dfrac{(\sum_{i\Equal1}^n a_i)^m}{\sum_{i\Equal1}^n x_i}&
\end{align}
\end{proof}

\begin{thm}\label{Theorem: sufficient condition3}
A sufficient condition for a $k$-variable first degree polynomial $P(x_1,...,x_k)= a_0 +\sum_{i\Equal 1}^k a_i x_i $ with real coefficients ($a_i\in \mathbb{R}$) is to have: Euclidean distance $n$ variables.
\begin{equation}
| a_i| \leq c/\sqrt[2]{k},\;\;\;\; \forall i=1...k 
\end{equation}
\end{thm}

\begin{proof}

For every two locations $\boldsymbol{l_1} = (x_1,...,x_k)$ and $\boldsymbol{l_2} = (x_1',...,x_k')$, a lower bound can be derived based on Generalized Titu's Lemma. 
\begin{equation}
    d (\boldsymbol{l_1},\boldsymbol{l_2})= \sqrt[2]{\sum_{i\Equal 1}^{k}(x_i-x_i')^2}\geq \sqrt[2]{(\sum_{i\Equal 1}^{k}|x_i-x_i'|)^2/k}
\end{equation}
On the other hand, the following inequality can be written for the polynomial.

\begin{align}
    |P(x_1,...,x_k)- P(x_1',...,x_k')| &= |\sum_{i\Equal 1}^{k}a_i(x_i-x_i')|\\
    &\leq \sum_{i\Equal 1}^{k}|a_i||(x_i-x_i')|  
\end{align}
By combining two equations the sufficiency condition can be derived from,
\begin{align}
     \sum_{i\Equal 1}^{k}|a_i||(x_i-x_i')|  \leq c \times (\sum_{i\Equal 1}^{k}|x_i-x_i'|)/\sqrt{k}\nonumber
\end{align}
\end{proof}

\section{Appendix (Polynomial Visualization)}\label{Section: appendix}

\revision{Figure~\ref{Fig: new fitting example} presents the real-world example of Fig.~\ref{Fig: Architecture} for NYC taxi dataset. It can be seen in the figure that a fair polynomial is fitted to output ML scores shown by blue dots. The degree of the polynomial is $6$, and the value of $c$ is set to $100$. }

\begin{figure}[t]
%\centering
\includegraphics[scale=.4]{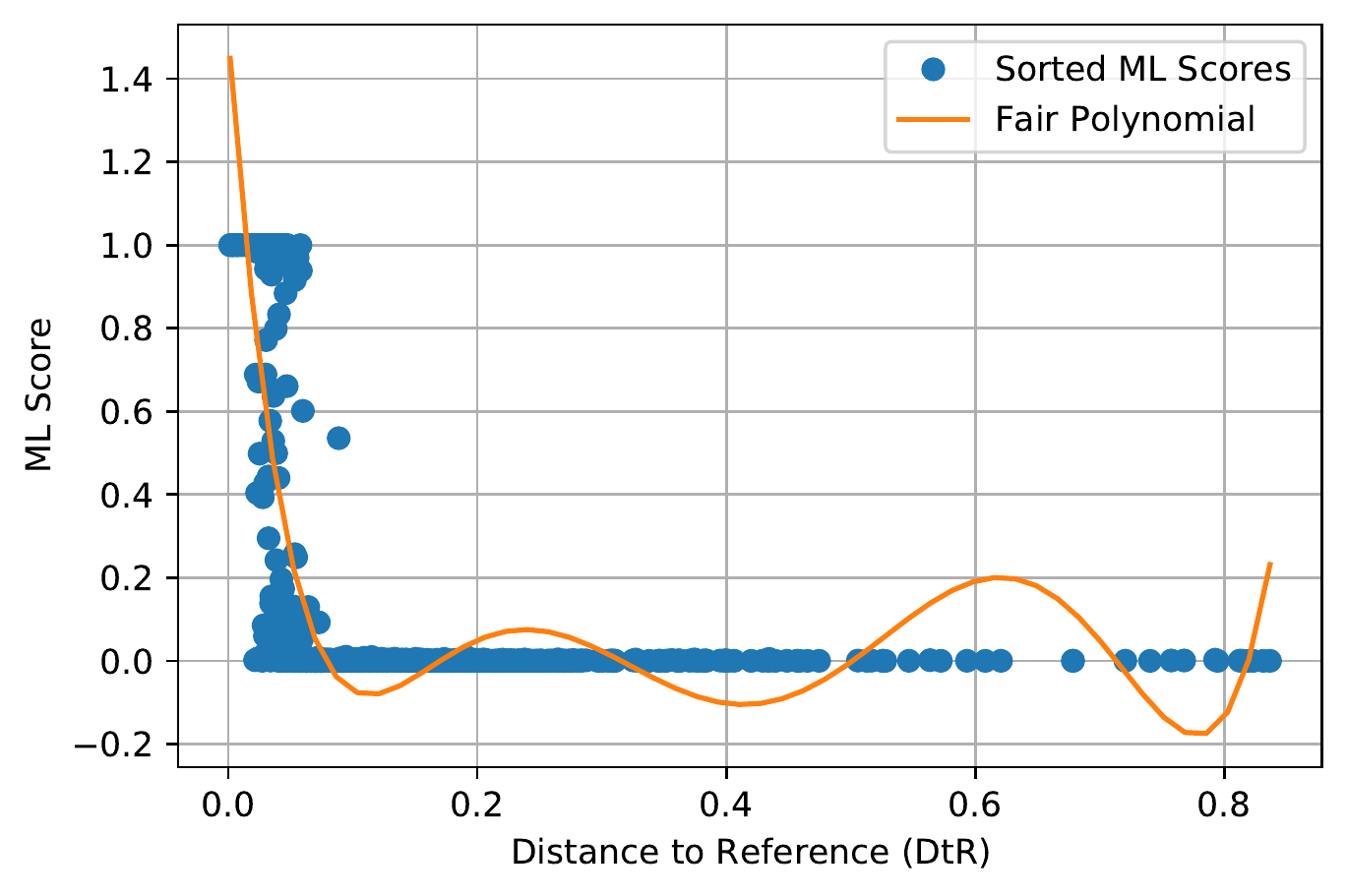}
%\hspace{1em}
\centering
%\vspace{-20pt}
\caption{\revision{An example of $c$-fair polynomial applied to the NYC taxi dataset.}}
%\vspace{-15pt}
\label{Fig: new fitting example}
\end{figure}

\end{document}